\documentclass[final,10pt]{elsarticle}
\usepackage{graphicx}
\usepackage{rotating}
\usepackage{epstopdf, epsfig}
\usepackage[utf8]{inputenc} 
\usepackage[T1]{fontenc}    
\usepackage{hyperref}       
\usepackage{url}            
\usepackage{amsfonts}       
\usepackage{amsthm}       
\usepackage{nicefrac}       
\usepackage{microtype}      
\usepackage{graphics}
\usepackage{pgfplots}
\usepackage{pgfplotstable}
\usepgfplotslibrary{fillbetween}
\usepackage{tikz}
\usepackage{bm}
\usepackage{xfrac}
\usepackage{float}
\usepackage{color}
\usepackage{subcaption}
\usepackage{algorithm}
\usepackage[noend]{algpseudocode}
\usepackage{mathtools}
\usepackage{amssymb}
\usepackage{amsmath}
\usepackage{standalone}
\usepackage{fullpage}
\usepackage[keeplastbox]{flushend}

\newcommand{\num}[1]{$#1$}
\newcommand{\subspace}{\mathcal S}
\newcommand{\projection}[1]{\boldsymbol P_{#1}}

\newcommand{\identity}{\boldsymbol I}
\newcommand{\nnFuncEncoderArg}[1]{\mathbf{h}_{#1}}
\newcommand{\nnFuncDecoderArg}[1]{\bar{\mathbf{h}}_{#1}}
\newcommand{\nlayersEncoder}{N_{\text{layers}}}
\newcommand{\nlayersDecoder}{\bar N_{\text{layers}}}
\newcommand{\weightNNEncoderarg}[1]{\boldsymbol{\Theta}_{#1}}
\newcommand{\weightNNDecoderarg}[1]{\bar{\boldsymbol{\Theta}}_{#1}}
\newcommand{\weightEncoder}{\boldsymbol{\theta}}
\newcommand{\weightEncoderOpt}{\weightEncoder^\star}
\newcommand{\weightDecoder}{\bar{\boldsymbol{\theta}}}
\newcommand{\weightDecoderOpt}{\weightDecoder^\star}
\newcommand{\weight}{\boldsymbol\theta}
\newcommand{\weightArg}[1]{\theta_{#1}}
\newcommand{\weightDummy}{\bar\weight}
\newcommand{\features}{\state}
\newcommand{\nfeatures}{\nstate}
\newcommand{\SVRfeaturespace}{\mathcal H}
\newcommand{\activationfunction}{h}
\newcommand{\nNeuronEncoder}[1]{p_{{#1}}}
\newcommand{\nNeuronDecoder}[1]{\bar p_{{#1}}}
\newcommand{\functionWithActivation}{g}
\newcommand{\regularization}{\Omega}

\newcommand{\stateSymb}{x}
\newcommand{\state}{\boldsymbol{\stateSymb}}
\newcommand{\stateArg}[1]{\state^{#1}}
\newcommand{\stateArgs}[2]{\stateSymb_{#1}^{#2}}
\newcommand{\stateApprox}{\tilde\state}
\newcommand{\stateApproxArg}[1]{\stateApprox^{#1}}

\newcommand{\nstate}{N_x}
\newcommand{\nstateArg}[1]{N_{x,#1}}

\newcommand{\velocitySymb}{f}
\newcommand{\velocity}{\boldsymbol{\velocitySymb}}

\newcommand{\velocityArg}[1]{\velocitySymb_{#1}}
\newcommand{\velocityApprox}{\tilde{\boldsymbol{f}}}
\newcommand{\velocityApproxArg}[1]{\tilde{f}_{#1}}
\newcommand{\regressionModel}{\tilde{f}}
\newcommand{\regressionModelArg}[1]{\tilde{f}^{#1}}
\newcommand{\nfeaturesSplit}{N_{\text{split}}}
\newcommand{\responseTrain}[1]{\bar y^{#1}}
\newcommand{\featuresTrain}[1]{{\bar{\boldsymbol x}^{#1}}}
\newcommand{\svrMapping}{\boldsymbol \psi}
\newcommand{\ntree}{{N_\text{tree}}}
\newcommand{\nearestSetArg}[1]{\mathcal{I}(#1)}
\newcommand{\knnWeight}{\tau}

\newcommand{\nSINDyterm}{{n_\text{SINDy}}}
\newcommand{\nVKOGA}{{n_\text{VKOGA}}}
\newcommand{\basisFunctionArg}[1]{g_{#1}}

\newcommand{\VKOGAbasis}{\boldsymbol{\alpha}}

\newcommand{\timeSymb}{t}
\newcommand{\card}[1]{|#1|}
\newcommand{\timeArgPCA}[1]{n_{\text{PCA},#1}}
\newcommand{\timeArgRegression}[1]{n_{\text{reg},#1}}
\newcommand{\ntimeTrainPCA}{{\card{\timeDomainTrainPCA}}}
\newcommand{\nPCA}{n_{\hat{\mathbf w}}}
\newcommand{\ntimeTrainRegression}{{\card{\timeDomainTrainRegression}}}
\newcommand{\ntrainingData}{{n_\mathrm{train}}}
\newcommand{\ntime}{{N_\timeSymb}}

\newcommand{\fullToRed}{\mathbf{r}}
\newcommand{\redToFull}{\mathbf{p}}
\newcommand{\encoderFunction}{\phi}
\newcommand{\leftSing}{\boldsymbol U}
\newcommand{\leftSingArg}[1]{\boldsymbol U_{#1}}
\newcommand{\leftSingVecArg}[2]{\boldsymbol u_{#1,#2}}
\newcommand{\leftSingVec}[1]{\boldsymbol u_{#1}}

\newcommand{\range}[1]{\text{Ran}(#1)}
\newcommand{\diag}[1]{\mathrm{diag}(#1)}
\newcommand{\Sing}{\boldsymbol \Sigma}
\newcommand{\SingArg}[1]{\boldsymbol \Sigma_{#1}}
\newcommand{\rightSing}{\boldsymbol V}
\newcommand{\rightSingArg}[1]{\boldsymbol V_{#1}}
\newcommand{\encoderFunctionGlobal}{\phi_\text{global}}
\newcommand{\decoderFunction}{\psi}
\newcommand{\decoderFunctionGlobal}{\psi_\text{global}}
\newcommand{\stateGlobalVector}{\boldsymbol w}

\newcommand{\stateGlobalVectorArg}[1]{\stateGlobalVector^{#1}}
\newcommand{\stateGlobalVectorApproxArg}[1]{\tilde{\boldsymbol w}^{#1}}
\newcommand{\nstateGlobalVector}{N_{\stateGlobalVector}}
\newcommand{\stateLocalVector}[1]{\boldsymbol w_{#1}}
\newcommand{\stateLocalVectorApprox}[1]{\tilde{\boldsymbol w}_{#1}}
\newcommand{\stateLocalVectorRed}[1]{\hat{\stateGlobalVector}_{#1}}
\newcommand{\stateGlobalVectorRed}{\hat{\stateGlobalVector}}
\newcommand{\stateGlobalVectorRedArg}[1]{\hat{\mathbf w}_{#1}}
\newcommand{\stateGlobalVectorRedAvg}{\bar{{\boldsymbol w}}}
\newcommand{\stateGlobalVectorRedAvgArg}[1]{\bar{{\mathbf w}}_{#1}}
\newcommand{\nstateLocalVector}{N_{{\stateGlobalVector},\text{el}}}
\newcommand{\nstateLocalVectorRed}{N_{\hat{\stateGlobalVector},\text{el}}}
\newcommand{\vectorize}[1]{\text{vec}(#1)}
\newcommand{\defeq}{:=}
\newcommand{\timeDomain}{\mathbb T}
\newcommand{\timeDomainGen}{\timeDomain}
\newcommand{\timeDomainSample}{\timeDomain_\text{sample}}
\newcommand{\timeDomainTrainAuto}{\timeDomain_\text{autoencoder}}
\newcommand{\timeDomainTrainPCA}{\timeDomain_\text{PCA}}
\newcommand{\timeDomainTrainRegression}{\timeDomain_\text{reg}}
\newcommand{\trainingDataArg}[1]{\mathcal T_\mathrm{train,#1}}
\newcommand{\trainingData}{\mathcal T_\mathrm{train}}
\newcommand{\PCA}{\boldsymbol \Phi}
\newcommand{\PCAArg}[1]{\boldsymbol \Phi_{#1}}

\newcommand{\nelements}{N_\text{el}}

\newcommand{\RR}[1]{\mathbb R^{#1}}

\newcommand{\innat}[1]{=1,\ldots,#1}

\newcommand{\zero}{{\mathbf{0}}}
\definecolor{colorGreen}{RGB}{0,130,0}

\newcommand{\R}{\mathbb{R}}

\providecommand{\e}[1]{\ensuremath{\times 10^{#1}}}

\newcommand{\Reynolds}{\mathrm{Re}}
\newcommand{\Prandtl}{\mathrm{Pr}}

\newcommand{\reviewer}[1]{#1}

\newcommand{\innerprod}[2]{\langle #1,#2\rangle}

\newcommand{\LSBoosting}{F}
\newcommand{\WeekLearner}{h}
\newcommand{\xspace}{X}
 
\newcommand{\idset}{\mathbb{I}}
\newcommand{\id}{i}
\newcommand{\maxcoeff}{\alpha}
\newcommand{\bignumber}{M}

\theoremstyle{definition}

\newtheorem{theorem}{Theorem}[section]
\newtheorem{corollary}{Corollary}[theorem]
\newtheorem{lemma}[theorem]{Lemma}

\begin{document}
\numberwithin{equation}{section}
\begin{frontmatter}
\title{Recovering missing CFD data for high-order discretizations using\\ deep
neural networks and dynamics learning}

	\author[sandia]{Kevin T.\ Carlberg\corref{sandiacor}\fnref{authors}}
\ead{ktcarlb@sandia.gov}
\ead[url]{sandia.gov/~ktcarlb}
\author[texas]{Antony Jameson\fnref{texascor}}
	\ead{antony.jameson@tamu.edu}
\ead[url]{engineering.tamu.edu/aerospace/profiles/jameson-antony.html}
	\author[stanford]{Mykel J.\  Kochenderfer\fnref{stanfordcor}}
\ead{mykel@stanford.edu}
\ead[url]{mykel.kochenderfer.com}
\author[stanford]{Jeremy Morton\fnref{stanfordcor}}
\ead{jmorton2@stanford.edu}
\author[sandia]{Liqian Peng\fnref{liqian}}
\ead{pengliqian@gmail.com}
\author[stanford]{Freddie D.\ Witherden\fnref{stanfordcor}}
\ead{fdw@stanford.edu}
\ead[url]{freddie.witherden.org}

\address[sandia]{Sandia National Laboratories}
\fntext[authors]{Authors are listed in alphabetical order.}
\cortext[sandiacor]{7011 East Ave, MS 9159, Livermore, CA 94550. Sandia National Laboratories is a multimission laboratory managed and operated
by National Technology \& Engineering Solutions of Sandia, LLC, a wholly owned
subsidiary of Honeywell International Inc., for the U.S. Department of
Energy's National Nuclear Security Administration under contract DE-NA0003525.}
\address[stanford]{Stanford University}
\fntext[stanfordcor]{Durand Building, 496 Lomita Mall, Stanford University, Stanford, CA 94305-3035.}
\address[texas]{Texas A\&M University}
\fntext[texascor]{701 H.R. Bright Bldg, College Station, TX 77843-3141.}

\begin{abstract}
Data I/O poses a significant bottleneck in large-scale CFD simulations; thus,
	practitioners would like to significantly reduce the number of times the
	solution is saved to disk, yet retain the ability to recover any field
	quantity (at any time instance) \textit{a posteriori}. 	The
	objective of this work is therefore to accurately recover missing CFD data
	\textit{a posteriori} at any time instance, given that the solution has been
	written to disk at only a relatively small number of time instances. We
	consider in particular high-order discretizations (e.g., discontinuous
	Galerkin), as such techniques are becoming increasingly popular for the
	simulation of highly separated flows.  To satisfy this objective, this work
	proposes a methodology consisting of two stages: 1) dimensionality reduction
	and 2) dynamics learning. For dimensionality reduction, we propose a novel
	hierarchical approach. First, the method reduces the number of degrees of
	freedom within each element of the high-order discretization by applying
	autoencoders from deep learning. Second, the methodology applies principal
	component analysis to compress the global vector of encodings. This leads to
	a low-dimensional state, which associates with a nonlinear embedding of the
	original CFD data. For dynamics learning, we propose to apply regression
	techniques (e.g., kernel methods) to learn the discrete-time velocity
	characterizing the time evolution of this low-dimensional state. A numerical
	example on a large-scale CFD example characterized by nearly 13 million
	degrees of freedom illustrates the suitability of the proposed method
	in an industrial setting.
\end{abstract}

\end{frontmatter}

\section{Introduction}
Industrial practitioners of
Computational Fluid Dynamics (CFD) often desire high-fidelity scale-resolving
simulations of transient compressible flows within the vicinity of complex
geometries.  For example, to improve the design of next-generation aircraft,
one must simulate---at Reynolds numbers
in the range
$10^4$--$10^7$ and Mach numbers in the range $0.1$--$1.0$---highly separated flow
over deployed spoilers/air-brakes; separated flow within serpentine intake
ducts; and flow over entire vehicle configurations at off-design conditions.
In order to perform these simulations, it is necessary to solve the unsteady
compressible Navier--Stokes equations with a high level of fidelity.

Theoretical studies and numerical experiments have shown high-order accurate
discontinuous spectral element methods (DSEM) to be particularly well suited
for these types of simulations \cite{vincent2011facilitating,
vermeire2017utility, park2017high}.  Indeed, we remark that some of the
largest CFD simulations to date---involving hundreds of billions of degrees of
freedom---have been performed using these methods \cite{vincent2016towards}.
DSEMs work by \reviewer{first} decomposing the domain into a set of conforming elements and
then \reviewer{representing} the solution inside of each element by a polynomial.
However, in contrast to classical finite element methods, the solution is
permitted to be discontinuous across element boundaries.   Examples of DSEMs
include discontinuous Galerkin (DG) methods
\cite{reed1973triangular,cockburn1991runge,hesthaven2008nodal}, spectral
difference (SD) schemes
\cite{kopriva1998staggered,liu2006spectral,sun2007high}, and flux
reconstruction (FR) schemes \cite{huynh2007flux}.

On account of the large number of spatial degrees of freedom (arising from the
large number of discretization elements and high-order polynomials within each
element) and small time steps that are often required for accuracy, such
simulations have the potential to generate petabytes of solution data; this is
well beyond what can be handled by current I/O and storage
sub-systems.  \reviewer{Therefore,} it is not practical to write out a sequence of finely
spaced solution snapshots to disk for offline analysis and post-processing.
Instead, the practitioner must instrument the simulation in advance.  Such
instrumentation typically includes the online accumulation of time averages of
various volumetric expressions, integrating forces on boundaries, and the
regular sampling of the flow field at specific points in the volume.  However,
doing this effectively requires a degree of \emph{a priori} knowledge about
the dynamics and evolution of the system\reviewer{, which} negates many of the
exploratory advantages inherent to simulation.

A common strategy for reducing the storage overhead associated with CFD
simulations is to compress the data before writing to disk
\cite{kang2003study, sakai2013parallel, trott1996wavelets}.  Although these
approaches can substantially reduce file sizes, they do not necessarily reduce
the overall number of time instances at which a file must be written. In this
paper, we instead consider the problem of accurately recovering the full CFD solution at any time instance \textit{a posteriori}, given that
it has been written at only a relatively small number of time instances.
Several methods have been proposed to reconstruct missing simulation data,
especially in the context of CFD.
However, most of these methods either aim to reconstruct missing spatial
data (which is not the focus of this work)
\cite{bui2004aerodynamic,willcox2006unsteady} or are not amenable to scenarios
where no data has been written at some time instances
\cite{beckers2003eof,venturi2004gappy,gunes2006gappy}.
 
 To address this work's objective, we propose a novel two-stage strategy that
 leverages recent innovations in machine learning: 1) 
 dimensionality reduction and 2) dynamics learning.
\textit{The first stage is dimensionality reduction}. We propose a novel
hierarchical approach to dimensionality reduction that leverages the structure
of high-order discretizations. It comprises \textit{local compression} using
autoencoders from deep learning followed by \textit{global compression} using principal component
analysis (PCA). This stage computes an accurate, low-dimensional, nonlinear
embedding of the original data while remaining computationally tractable for
large-scale data.

For local compression, we reduce the large (e.g., $\sim 10^2$) number of degrees of freedom per element in the
mesh via autoencoders; this step leverages the structure of high-order
discretizations. An autoencoder is a (typically deep) neural network that performs
dimensionality reduction by learning a nonlinear mapping from
high-dimensional feature space to a lower-dimensional
encoding~\cite{hinton2006reducing}.  A standard autoencoder consists of two
components: an \textit{encoder} network that performs the mapping between input
data and the encodings and \reviewer{a} \textit{decoder} network that attempts to reconstruct the
original input from the encodings.  The autoencoder and its variants have
previously been used extensively for discovering low-dimensional
representations of image data~\cite{vincent2010stacked,kingma2013autoencoding,
theis2017lossy}.
However, rather than targeting image data, we use autoencoders for
local compression of CFD data.  Recent work has explored using autoencoder
architectures for discovering low-dimensional representations of fluid flow,
primarily for the purposes of simulation and flow
control~\cite{wiewel2018latent, otto2017linearly, lusch2017deep,puligilla2018deep,
morton2018deep}.  In these studies, \textit{entire solution states}
could be treated as neural network inputs (i.e., features) due to the
relatively small number of degrees of freedom in the systems being simulated.

Unfortunately, memory constraints prohibit training autoencoders on full
solution states for many large-scale CFD simulations.  Furthermore, due
to the large number of parameters that must be optimized, neural networks
typically require many training examples; this work assumes that the solution
\reviewer{has} been written at only a relatively small number of time instances.  The proposed method avoids these issues: the features
correspond to the degrees of freedom within a single element (and so pose no memory
issue), and the number of training examples is the number of time instances at
which the solution has been written \textit{multiplied by} the number of
elements in the mesh (which may be very large, e.g., ${\sim} 10^6$).

Global compression is necessary because even if autoencoders significantly reduce the number of
degrees of freedom in each element, the number of elements in the mesh
may still be large. We therefore apply classical principal
component analysis (with some modifications to handle large-scale data) to
reduce the dimensionality of the vector of local encodings across the entire
spatial domain. Although  PCA identifies a linear subspace, PCA in combination with
autoencoders yields a nonlinear embedding of the global CFD solution. In the
past, PCA has been widely applied to global data in the context of CFD, but
often under the name proper orthogonal decomposition (POD)
\cite{POD,carlbergGalDiscOpt}.  However, to our knowledge, PCA has not yet
been applied to a vector of autoencoder codes.

\textit{The second stage is dynamics learning}. Because this work assumes that
the solution has been written to disk at only a relatively small subset of the desired time
instances, we must devise an approach for reconstructing the solution at the
missing time instances. For this purpose, we propose learning the dynamics. 
In particular, we assume that there exists a Markovian
discrete-time dynamical system that describes the time-evolution of the
low-dimensional latent state (obtained after applying both local and global
compression) on the time grid of interest; of course, such a dynamical system
may not exist due to the closure problem. Next, we aim to approximate the
discrete-time velocity characterizing this hypothesized dynamical system. We
regress the (possibly nonlinear) mapping from the low-dimensional latent state at the current time
instance to the low-dimensional latent state at the next time
instance. For this purpose, we investigate the viability of a wide range of
regression techniques within a machine-learning framework for model selection, including support vector regression \cite{smola_2004},
random forests \cite{breiman_2001}, boosted decision trees \cite{Hastie2009},
$k$-nearest neighbors \cite{altman1992introduction}, the vectorial kernel
orthogonal greedy algorithm (VKOGA)
\cite{wirtz2013vectorial,wirtz2015surrogate}, sparse identification of
nonlinear dynamics (SINDy) \cite{brunton2016discovering}, and dynamic mode
decomposition (DMD) \cite{schmid2010dynamic}. 

Dynamics learning is an active area of research, and many approaches have been
proposed for learning dynamics in multiple fields. However, none has been
applied for the purpose of reconstructing missing simulation data, and few of
these methods are applicable to the kind of data we consider: very large-scale
states, and relatively few time instances at which the state has been
recorded. For example, the field of data-driven dynamical systems aims to
learn (typically continuous-time) dynamics from observations of the 
state or velocity.  Many of these methods enforce linear dynamics: DMD
\cite{schmid2010dynamic} enforces linear dynamics for the full-system state,
while methods based on Koopman theory enforce linear dynamics on a finite
number of system observables, which can be specified either manually
\cite{williams2014kernel,williams2015data,kawahara2016dynamic} or learned
using autoencoders
\cite{takeishi2017learning,otto2017linearly,lusch2017deep,morton2018deep}.
Other works have enabled nonlinear dynamics, but a rich library of candidate
basis functions must be manually specified \cite{brunton2016discovering}, and
model selection using validation data has been limited. In contrast to these
methods, our work aims to model nonlinear dynamics with a large candidate set of
functional forms for the velocity, and subsequently performs model selection using
validation data. We also emphasize that the aforementioned autoencoder
approaches are not applicable to the large-scale data sets we consider, as the
entire state is treated as an input to the autoencoder.

Relatedly, the field of state representation learning
\cite{lesort2018state,bohmer2015autonomous} aims to learn simultaneously both an underlying
low-dimensional latent space and discrete-time dynamics on this latent space
from a set of high-dimensional observations (e.g., raw pixels from a camera),
typically in a control or reinforcement-learning context.  The dynamics are typically constrained to be (locally) linear to
facilitate control \cite{NIPS2015_5951,watter2015embed,karl2016deep}. In
contrast to these approaches, we are not interested in control and thus need
not restrict the dynamics to be linear. Finally, these methods would encounter
difficulties when applied to the data sets we consider (i.e., a
high-dimensional observation space, but data at relatively few time
instances), as simultaneously learning an autoencoder and dynamics on an
observation space of dimension $\sim 10^7$ requires a substantial amount of
computation and data.


Another contribution of this work is the application of the proposed methodology to
a large-scale industrial CFD application characterized by a high-order
discretization with $320$ degrees of freedom per element and $40\,584$
elements, yielding nearly $13\times 10^6$ degrees of freedom. These numerical
experiments demonstrate the viability of the proposed method, as it 
significantly reduces the dimension of the solution to only 500 (for a 
compression ratio of $26\,000:1$), yet 
accurately recovers CFD data, as it incurs sub-1\% relative errors in the global
solution and drag over all desired $8\,600$ time instances. The numerical
experiments show that the VKOGA algorithm yields the best performance for
dynamics learning, due both to its low regression test error and its boundedness. In
contrast, while SINDy yielded low test regression error, it yielded an
unstable dynamical system to the unboundedness of the polynomial basis
functions it employs.

The remainder of the paper is organized as follows. Section
\ref{sec:motivation} provides the problem formulation. Section
\ref{sec:dimRed} describes the two-stage dimensionality-reduction process;
here, Section \ref{sec:autoencoders} describes local compression using
autoencoders, and Section \ref{sec:PCA} describes global compression using PCA.
Section \ref{sec:regressionDiscrete} describes dynamics learning, wherein
Section \ref{sec:regressionSetting} precisely describes the regression
setting, Section \ref{sec:training} describes the required training data,  and
Section \ref{sec:regressionModels} provides a description of all considered
regression models. Section \ref{sec:boundedness} provides some analysis
related to boundedness of the learned discrete-time velocity, which has
implications on the stability of the resulting dynamical-system model. Section
\ref{sec:experiments} reports the numerical experiments. Finally, Section
\ref{sec:conclusions} concludes the paper.

\section{Problem formulation: recovering a state sequence for high-order
discretizations}\label{sec:motivation}

The objective of this work is to recover a sequence of states
$
	\{\stateGlobalVectorArg{n}\}_{n=0}^\ntime\subset\RR{\nstateGlobalVector}
	$
given only the initial state
$\stateGlobalVectorArg{0}\in\RR{\nstateGlobalVector}$ and a subset of the
elements of the sequence
$\{\stateGlobalVectorArg{n}\}_{n\in\timeDomainSample}$ with
$\timeDomainSample\subset\{0,\ldots,\ntime\}$. This scenario often arises in
CFD applications, where the practitioner is interested in having access to the fluid solution
at many time instances, but it is prohibitively expensive to write the
solution to disk for all time instances of interest. In this case, the state corresponds to the vector of
fluid variables over the spatial domain, and the
sequence corresponds to the value of these degrees of freedom at time
instances of interest (which need not correspond to the time steps taken by
the time integrator).

This work aims to satisfy this objective in the context of high-order spatial
discretizations (e.g., discontinuous Galerkin) of partial differential
equations (PDEs), which also often arise in CFD applications. Such discretizations are characterized by $\nelements$ elements, each
of which has $\nstateLocalVector\gg 1$ local degrees of freedom. Mathematically, we can
characterize the state as the vectorization of local states as
 \begin{equation} 
	 \stateGlobalVector\defeq\vectorize{\stateLocalVector{1},\ldots,\stateLocalVector{\nelements}},
	  \end{equation} 
		where $\stateLocalVector{i}\in\RR{\nstateLocalVector}$, $i\innat{\nelements}$
		denotes the local state associated with the $i$th element, and 
		$\stateGlobalVector\in\RR{\nstateGlobalVector}$ denotes the global state
		with $\nstateGlobalVector = \nelements\nstateLocalVector$.

To achieve the stated objective, we adopt a two-stage process:
\begin{enumerate} 
	\item\label{step:dimred} \textit{Dimensionality reduction}. We apply
		hierarchical dimensionality reduction
		to the state $\stateGlobalVector$ in two steps.
		First,
we apply autoencoders to reduce the local dimensionality, i.e., the
dimensionality of the local states $\stateLocalVector{i}$, $i\innat{\nelements}$ (Section \ref{sec:autoencoders}). Second, we
apply principal component analysis (PCA) to globally reduce the dimensionality
of the set of autoencoder-compressed local states (Section \ref{sec:PCA}).
		This
		results in a low-dimensional state $\state =
		\fullToRed(\stateGlobalVector)\in\RR{\nstate}$ with
		$\nstate\ll\nstateGlobalVector$ from
		which the full state can be approximated as
		$\stateGlobalVector\approx
		\redToFull(\state)\in\RR{\nstateGlobalVector}$. Here,
		$\fullToRed:\RR{\nstateGlobalVector}\rightarrow\RR{\nstate}$ is the
		restriction operator associated with dimensionality reduction, and 
		$\redToFull:\RR{\nstate}\rightarrow\RR{\nstateGlobalVector}$ is the
		prolongation operator.
	\item\label{step:dynamicslearn} \textit{Dynamics learning}. We regress the discrete-time
		dynamics of the low-dimensional state $\state$ (Section
		\ref{sec:regressionDiscrete}) under the assumption that the sequence of
		low-dimensional states
		$ \{\stateArg{n}\}_{n=0}^\ntime\subset\RR{\nstate} $ can be recovered from
		the Markovian dynamical system
 \begin{equation} 
\stateArg{n+1} = \velocity(\stateArg{n}),\quad n=0,\ldots,\ntime-1
	  \end{equation} 
		with $\stateArg{0} = \fullToRed(\stateGlobalVectorArg{0})$
		for some unknown velocity $\velocity$. Then, we construct a regression
		approximation to the discrete-time velocity
		$\velocityApprox\approx\velocity$. This allows the desired sequence
		$\{\stateGlobalVectorArg{n}\}_{n=0}^\ntime$
		to be approximated as 
		$\{\stateGlobalVectorApproxArg{n}\}_{n=0}^\ntime$,
		where
		$\stateGlobalVectorApproxArg{0} = \stateGlobalVectorArg{0}$ and
		$\stateGlobalVectorApproxArg{n} = \redToFull(\stateApproxArg{n})$,
		$n=1,\ldots,\ntime$. Here, the approximated low-dimensional state evolves
		according to the approximated dynamics
 \begin{equation} 
	 \stateApproxArg{n+1} = \velocityApprox(\stateApproxArg{n}),\quad
		n=0,\ldots,\ntime-1
	  \end{equation} 
		with $\stateApproxArg{0} =
		\fullToRed(\stateGlobalVectorArg{0})$.
\end{enumerate}

We note that while stage 2 could be executed without stage 1,
learning discrete-time dynamics of a low-dimensional state requires
substantially less data and computational effort than doing so for a
high-dimensional state.



\section{Hierarchical dimensionality reduction for high-order
discretizations}\label{sec:dimRed}
This section describes the first stage of the proposed methodology:
dimensionality reduction.  Section \ref{sec:autoencoders} describes local
compression using autoencoders, and Section \ref{sec:PCA} describes global
compression using PCA.

\subsection{Local compression using autoencoders}\label{sec:autoencoders}

The first step of the proposed approach is to reduce the dimensionality of the
local states.  To achieve this, we make the observation that the total number
of available training examples is equal to the number of samples
$\card{\timeDomainSample}$ multiplied by the number of elements $\nelements$,
where $\card{\cdot}$ denotes the cardinality of a set. For fine spatial
discretizations, $\nelements$ is large, which implies that we may have access
to a large amount of training data, even for a single simulation with
$\card{\timeDomainSample}\sim 100$. Due to this fact, we apply autoencoders,
as they often enable very accurate nonlinear embeddings, yet
require a large amount of training data.

In the present context, autoencoders aim to find two mappings:  the
\textit{encoder} $\encoderFunction: \RR{\nstateLocalVector} \rightarrow
\RR{\nstateLocalVectorRed}$ that maps a local state
$\stateLocalVector{i}\in\RR{\nstateLocalVector}$ to lower-dimensional
representation $\stateLocalVectorRed{i}\in\RR{\nstateLocalVectorRed}$; i.e.
$\encoderFunction:\stateLocalVector{i}\mapsto\stateLocalVectorRed{i}$,
$i\innat{\nelements}$; and the \textit{decoder} $\decoderFunction:
\RR{\nstateLocalVectorRed}\rightarrow\RR{\nstateLocalVector}$.  Here,
$\nstateLocalVectorRed( \ll \nstateLocalVector)$ denotes the reduced dimension
of the local state vector or \textit{code}. Its value dictates the tradeoff
between compression and reconstruction accuracy.

The encoder $\encoderFunction$ takes the form of a feed-forward
neural network with $\nlayersEncoder$ layers constructed such that
\begin{equation} 
	\encoderFunction(\stateLocalVector{i};\weightEncoder)=\nnFuncEncoderArg{\nlayersEncoder}(\cdot;\weightNNEncoderarg{\nlayersEncoder})\circ
  \nnFuncEncoderArg{\nlayersEncoder-1}(\cdot;\weightNNEncoderarg{\nlayersEncoder-1})\circ
  \cdots\circ\nnFuncEncoderArg{1}(\stateLocalVector{i};\weightNNEncoderarg{1}),
\end{equation} 
where
$\nnFuncEncoderArg{i}(\cdot;\weightNNEncoderarg{i}):\RR{\nNeuronEncoder{i-1}}\rightarrow\RR{\nNeuronEncoder{i}}$,
$i=1,\ldots, \nlayersEncoder$
denotes the function applied at layer $i$ of the neural network;
$\weightNNEncoderarg{i}$, $i=1,\ldots,\nlayersEncoder$ denote the weights
employed at
layer $i$; $\nNeuronEncoder{i}$ denotes the dimensionality of the output at layer $i$;
and $\weightEncoder\equiv({\weightNNEncoderarg{1}},\ \ldots,\ 
{\weightNNEncoderarg{\nlayersEncoder}})$.  The
input is of dimension
$\nNeuronEncoder{0} = \nstateLocalVector$ and the final (output) layer
($i=\nlayersEncoder$)
produces the code
$\stateLocalVectorRed{i} = \encoderFunction(\stateLocalVector{i})$ such that
$\nNeuronEncoder{\nlayersEncoder}=\nstateLocalVectorRed$. An activation function $\activationfunction$ is applied in layers 1
to $\nlayersEncoder$ to
some function of the weights and the outputs from the previous layer
$\mathbf{y}_{i-1}\in\RR{\nNeuronEncoder{i-1}}$
such that
\begin{equation}
  \nnFuncEncoderArg{i}(\mathbf{y}_{i-1};\weightNNEncoderarg{i}) =
		h(\functionWithActivation(\weightNNEncoderarg{i},\mathbf{y}_{i-1})),
\end{equation}
where
$\functionWithActivation(\weightNNEncoderarg{i},\mathbf{y}_{i-1})=\weightNNEncoderarg{i}[1;\mathbf{y}_{i-1}]$
with $\weightNNEncoderarg{i}$ a real-valued matrix for a
traditional multilayer perceptron (MLP), while
$\functionWithActivation$ corresponds to a convolution operator for a
convolutional neural network with $\weightNNEncoderarg{i}$ providing the
convolutional-filter weights.
Note that $\mathbf{y}_{0}\defeq\stateLocalVector{i}$. The activation
function is applied element-wise to the vector argument; common choices
include the rectified linear unit (ReLU), the logistic sigmoid, and the
hyperbolic tangent.

Analogously, the decoder $\decoderFunction$ also takes the form of a
feed-forward artificial neural network with $\nlayersDecoder$ such that
\begin{equation} 
	\decoderFunction(\stateLocalVector{i};\weightDecoder)=\nnFuncDecoderArg{\nlayersDecoder}(\cdot;\weightNNDecoderarg{\nlayersDecoder})\circ
  \nnFuncDecoderArg{\nlayersDecoder-1}(\cdot;\weightNNDecoderarg{\nlayersDecoder-1})\circ
  \cdots\circ\nnFuncDecoderArg{1}(\stateLocalVector{i};\weightNNDecoderarg{1}),
\end{equation} 
with 
$\nnFuncDecoderArg{i}(\cdot;\weightNNDecoderarg{i}):\RR{\nNeuronDecoder{i-1}}\rightarrow\RR{\nNeuronDecoder{i}}$,
$i=1,\ldots, \nlayersDecoder$,
and $\weightDecoder\equiv[\vectorize{\weightNNDecoderarg{1}};\ \ldots\ ;
\vectorize{\weightNNDecoderarg{\nlayersDecoder}}]$. The
input is of dimension
$\nNeuronDecoder{0} = \nstateLocalVectorRed$  and the final (output) layer
($i=\nlayersDecoder$)
produces the high-dimensional representation
$\stateLocalVectorApprox{i} = \decoderFunction(\stateLocalVectorRed{i})$ such that
$\stateLocalVectorApprox{i}\approx\stateLocalVector{i}$ and
$\nNeuronDecoder{\nlayersDecoder}=\nstateLocalVector$. As before,
\begin{equation}
  \nnFuncDecoderArg{i}(\mathbf{y}_{i-1};\weightNNDecoderarg{i}) =
		h(\functionWithActivation(\weightNNDecoderarg{i},\mathbf{y}_{i-1})),
\end{equation}
where $\mathbf{y}_{i-1}\in\RR{\nNeuronDecoder{i-1}}$ is the output from the previous
layer, and $\mathbf{y}_{0}\defeq\stateLocalVectorRed{i}$.

To train the autoencoder, we use data corresponding to the local states at
a subset of time instances $\timeDomainTrainAuto\subseteq\timeDomainSample$,
i.e., the training data corresponds to
$$
\{\stateLocalVector{i}^n\}_{i=1,\ldots,\nelements}^
{{n}\in\timeDomainTrainAuto}.
$$
We then compute the weights $(\weightEncoderOpt,\weightDecoderOpt)$ by solving
the minimization problem
\begin{equation}\label{eq:autoencoderOpt}
	\underset{(\weightEncoder,\weightDecoder)}{\mathrm{minimize}}\sum_{n\in\timeDomainTrainAuto}\sum_{i=1}^{\nelements}\|
	\stateLocalVector{i}^n -
	\decoderFunction(\cdot;\weightDecoder)\circ\encoderFunction(\stateLocalVector{i}^n;\weightEncoder)\|_2^2+\regularization(\weightEncoder,\weightDecoder),
\end{equation}
where $\regularization$ denotes a regularization function, e.g., ridge, lasso,
elastic-net~\cite{zou2005regularization}.
We note that this optimization problem is often solved approximately to
improve generalization behavior; for instance, 
optimization iterations are often terminated when the error on an independent
validation set begins to increase; see Ref.~\cite{bottou2018optimization} for
a review of training deep neural networks.
The encoder and decoder are then set to $\encoderFunction =
\encoderFunction(\cdot;\weightEncoderOpt)$ and $\decoderFunction =
\decoderFunction(\cdot;\weightDecoderOpt)$, respectively.

After applying the encoder, the global reduced state takes the form
 \begin{equation} 
	 \stateGlobalVectorRed\defeq\vectorize{\stateLocalVectorRed{1},\ldots,\stateLocalVectorRed{\nelements}}\in\RR{\nelements\nstateLocalVectorRed},
	  \end{equation} 
		where $\stateLocalVectorRed{i} = \encoderFunction(\stateLocalVector{i})$,
		$i=1,\ldots,\nelements$. 
		Note that we can also write $\stateGlobalVectorRed =
		\encoderFunctionGlobal(\stateGlobalVector)$, where
 \begin{align} 
 \begin{split} 
	 \encoderFunctionGlobal&:\stateGlobalVector\mapsto
	 \vectorize{\encoderFunction(\stateLocalVector{1}),\ldots,\encoderFunction(\stateLocalVector{\nelements})}\\
	 &:\RR{\nstateGlobalVector}\rightarrow\RR{\nelements\nstateLocalVectorRed}.
	  \end{split} 
	  \end{align} 
		The dimension of the global reduced state
		$\stateGlobalVectorRed$ is $\nelements\nstateLocalVectorRed$; while this is
		significantly smaller than the dimension 
		$\nstateGlobalVector = \nelements\nstateLocalVector$
		of the full state $\stateGlobalVector$, it may still be large if the
		number of elements $\nelements$ is large, as is the case for fine spatial
		discretizations. To address this, Section \ref{sec:PCA} describes how PCA
		can be applied to reduce the
		dimensionality of the global reduced state $\stateGlobalVectorRed$.

\subsection{Global compression using PCA}\label{sec:PCA}

Because the dimensionality of the global reduced state $\stateGlobalVectorRed
=
\encoderFunctionGlobal(\stateGlobalVector)\in\RR{\nelements\nstateLocalVectorRed}$
may still be large, we proceed by applying dimensionality reduction to this
global quantity. However, in this case we have many fewer training examples
than in the case of local compression. This arises from the fact that local
compression employs a training set whose cardinality is the product of
the number of training time instances and the number of elements (which is often
large). The size of the global-compression training set is simply the
number of training time instances.  Thus, for global compression, we use PCA, a
dimensionality reduction method that does not rely on access to a large amount
of data. Furthermore, because we are
usually considering large-scale data, it may be computationally costly to
compute global principal components to represent the global reduced state
vector $\stateGlobalVectorRed$; thus, we consider piecewise PCA for this
purpose.

To achieve this, we first 
define a set of training instances
$\timeDomainTrainPCA\equiv\{\timeArgPCA{j}\}_{j}
\subseteq\timeDomainSample$ and associated
global reduced states
$\{\stateGlobalVectorRed^n\}^{n\in\timeDomainTrainPCA}$ that will be employed to compute the principal components.  
Then, we decompose the global reduced state into $\nPCA$
components
$\stateGlobalVectorRedArg{i}$, $i=1,\ldots,\nPCA$ such that 
$\stateGlobalVectorRed\equiv[\stateGlobalVectorRedArg{1}^T\ \cdots\
\stateGlobalVectorRedArg{\nPCA}^T]^T$.
Next, we compute the singular value decompositions
\begin{equation} 
	[(\stateGlobalVectorRedArg{i}^{\timeArgPCA{1}}-\stateGlobalVectorRedAvgArg{i})\ \cdots\
	(\stateGlobalVectorRedArg{i}^{\timeArgPCA{\ntimeTrainPCA}}-\stateGlobalVectorRedAvgArg{i})] =
	\leftSingArg{i}\SingArg{i}\rightSingArg{i}^T, \quad i=1,\ldots,\nPCA,
\end{equation} 
where 
\begin{equation}
	\stateGlobalVectorRedAvgArg{i}\defeq
	\frac{1}{\ntimeTrainPCA}\sum_{n\in\timeDomainTrainPCA}\stateGlobalVectorRedArg{i}^n, \quad i=1,\ldots,\nPCA
\end{equation}
denote the sample means, and 
$\leftSingArg{i} \equiv[\leftSingVecArg{i}{1}\ \cdots\ \leftSingVecArg{i}{\ntimeTrainPCA}]$,
$\SingArg{i}= \diag{\sigma_{i,j}}_{j=1}^{\ntimeTrainPCA}$ with 
\begin{equation}
\sigma_{i,1}\geq\cdots\sigma_{i,\ntimeTrainPCA}\geq 0.
\end{equation}
Subsequently, PCA truncates the left singular vectors to obtain the basis matrix
 \begin{equation} 
	 \PCAArg{i} = \left[\leftSingVecArg{i}{1}\ \cdots\
	 \leftSingVecArg{i}{\nstateArg{i}}\right],
	  \end{equation} 
		where $\nstateArg{i}$ can be set according to a statistical energy
		criterion, for example.
We note that the resulting basis matrix satisfies the minimization problem
\begin{equation}
	\range{\PCAArg{i}} = \underset{\subspace,\, \dim(\subspace) = \nstateArg{i}}{\arg\min}
	\sum_{j=1}^\ntimeTrainPCA\|
	(\identity -
	\projection{\subspace})(\stateGlobalVectorRedArg{i}^{\timeArgPCA{j}}-\stateGlobalVectorRedAvgArg{i})
	\|_2^2,
\end{equation}
where $\projection{\subspace}$ denotes the orthogonal 
projector (in the Euclidean norm) onto the subspace $\subspace$; and 
minimization is taken over all subspaces of dimension $\nstateArg{i}$ (i.e.,
the Grassmannian).

We then compute the singular value decomposition
\begin{equation} \label{eq:SVDGlobal}
	[\diag{\PCAArg{i}}]^T[(\stateGlobalVectorRed^{\timeArgPCA{1}}-\stateGlobalVectorRedAvg)\ \cdots\
	(\stateGlobalVectorRed^{\timeArgPCA{\ntimeTrainPCA}}-\stateGlobalVectorRedAvg)] =
	\leftSing\Sing\rightSing^T,
\end{equation} 
with 
$\leftSing \equiv[\leftSingVec{1}\ \cdots\ \leftSingVec{\ntimeTrainPCA}]$
and
\begin{equation}
	\stateGlobalVectorRedAvg\defeq
	\frac{1}{\ntimeTrainPCA}\sum_{n\in\timeDomainTrainPCA}\stateGlobalVectorRed^n.
\end{equation}

We define the global basis matrix  
$\PCA\defeq\diag{\PCAArg{i}}\left[\leftSingVec{1}\ \cdots\
	 \leftSingVec{\nstate}\right]\in\RR{\nelements\nstateLocalVectorRed\times
\nstate}$ with
$\nstate\leq\sum_i\nstateArg{i}\ll\nelements\nstateLocalVectorRed\ll\nelements\nstateLocalVector$
determined from an energy criterion, for example. 

Now, we can define the reduced global state:
\begin{equation} \label{eq:redStateFromGlobalState}
	\stateArg{n} = \fullToRed(\stateGlobalVectorArg{n})
	,\quad
	n=0,\ldots,\ntime,
\end{equation} 
where $\fullToRed:\stateGlobalVector\mapsto 
\PCA^T(\encoderFunctionGlobal(\stateGlobalVector)-\stateGlobalVectorRedAvgArg{i})$.
The mapping from reduced state to the approximated global state is then
$\redToFull:\state\mapsto\decoderFunctionGlobal(\PCA\state+\stateGlobalVectorRedAvgArg{i})$,
where
 \begin{align} 
 \begin{split} 
	 \decoderFunctionGlobal&:\stateGlobalVectorRed\mapsto
	 \vectorize{\decoderFunction(\stateLocalVectorRed{1}),\ldots,\decoderFunction(\stateLocalVectorRed{\nelements})}\\
	 &:\RR{\nelements\nstateLocalVectorRed}\rightarrow\RR{\nstateGlobalVector}.
	  \end{split} 
	  \end{align} 

\section{Dynamics learning using regression}\label{sec:regressionDiscrete}

This section describes dynamics
learning.  In particular, we assume that the sequence of low-dimensional
states $ \{\stateArg{n}\}_{n=0}^\ntime\subset\RR{\nstate} $ with
$\stateArg{n}$ defined in Eq.~\eqref{eq:redStateFromGlobalState} can be
recovered from the Markovian discrete-time dynamical system
\begin{equation}\label{eq:discTimeState}
	\stateArg{n+1} = \velocity(\stateArg{n}),\quad n=0,\ldots,\ntime-1
\end{equation}
for some (unknown) discrete-time velocity
$\velocity:\RR{\nstate}\rightarrow\RR{\nstate}$. 
Then, if such a velocity exists and can be computed, the entire sequence of
low-dimensional states $ \{\stateArg{n}\}_{n=0}^\ntime\subset\RR{\nstate} $
can be recovered from Eq.~\eqref{eq:discTimeState} by specifying the initial
state $\stateArg{0}$ only; this allows an approximation to the sequence of
states to be computed as $ \{\stateGlobalVectorArg{n}\}_{n=1}^\ntime\approx
\{\redToFull(\stateArg{n})\}_{n=1}^\ntime$ (note that
$\stateGlobalVectorArg{0}$ is given).

Because we do not have direct access to this discrete-time velocity
$\velocity$, we must generate an approximation
$\velocityApprox\approx\velocity$ with
$\velocityApprox:\RR{\nstate}\rightarrow\RR{\nstate}$.  This work
constructs this approximation by regressing each component of the velocity.
Once the approximated velocity $\velocityApprox$ is constructed, the sequence
of low-dimensional states $ \{\stateArg{n}\}_{n=0}^\ntime\subset\RR{\nstate} $
can be approximated as
$ \{\stateApproxArg{n}\}_{n=0}^\ntime\subset\RR{\nstate}$, where 
$\stateApproxArg{0}=\stateArg{0}$ and $\stateApproxArg{n}$,
$n=1,\ldots,\ntime$ satisfies the approximated discrete-time dynamics 
\begin{equation}\label{eq:discTimeApprox}
\stateApproxArg{n+1} = \velocityApprox(\stateApproxArg{n}),\quad
	n=0,\ldots,\ntime-1.
\end{equation}
Then, the desired sequence of states
$\{\stateGlobalVectorArg{n}\}_{n=0}^\ntime$ can be approximated as
$\{\stateGlobalVectorApproxArg{n}\}_{n=0}^\ntime$, where
$\stateGlobalVectorApproxArg{0} = \stateGlobalVectorArg{0}$ and
$\stateGlobalVectorApproxArg{n} = \redToFull(\stateApproxArg{n})$,
$n=1,\ldots,\ntime$. 

\subsection{Regression setting}\label{sec:regressionSetting}
To cast this as a regression problem, we consider the $i$th equation of
\eqref{eq:discTimeState}:
\begin{equation} 
	\stateArgs{i}{n+1} = \velocityArg{i}(\stateArg{n}), \quad
	n=0,\ldots,\ntime-1,\quad i=1,\ldots,\nstate,
\end{equation} 
where $\state\equiv[\stateArgs{1}{}\ \cdots\
\stateArgs{\nstate}{}]^T$ and
$\velocity\equiv[\velocityArg{1}\ \cdots\ \velocityArg{\nstate}]^T$.
This expression shows that the $i$th component of the velocity
$\velocityArg{i}$ can be interpreted
as a function that maps the reduced state at the current time step
$\stateArg{n}$ to the $i$th element of the state at the next time step
$\stateArgs{i}{n+1}$ for $n=0,\ldots,\ntime-1$.

Based on this observation, we construct independent regression
models $\velocityApproxArg{i}\approx\velocityArg{i}$, $i=1,\ldots,\nstate$
whose \textit{features} (i.e., regression-model inputs) correspond to the
state at the current time step $\stateArg{n}$, and whose \textit{response}
(i.e., regression-model output) corresponds to the $i$th element of the state
at the next time step $\stateArgs{i}{n+1}$  for $n=0,\ldots,\ntime-1$.
The approximated velocity is then
$\velocityApprox\equiv[\velocityApproxArg{1}\ \cdots\
\velocityApproxArg{\nstate}]^T$.
Section \ref{sec:training} describes the training data used to construct the
regression models, and Section \ref{sec:regressionModels} describes the
proposed regression models.

\subsection{Training data}\label{sec:training}

Based on the regression setting described in Section
\ref{sec:regressionSetting}, it is clear that the appropriate training data
for constructing the $i$th regression model $\velocityApproxArg{i}$
corresponds to response--feature pairs
$$
\trainingDataArg{i}\defeq\{(\stateArgs{i}{n+1},\stateArg{n})\}_{n\in\timeDomainTrainRegression},
$$
where
$\timeDomainTrainRegression\equiv\{\timeArgRegression{j}\}_j\subset\timeDomainSample$
satisfies $\timeArgRegression{j}+1\in\timeDomainSample$,
$j=1,\ldots,\ntrainingData$, where
$\ntrainingData\defeq\ntimeTrainRegression=\card{\trainingDataArg{i}}$,
$i=1,\ldots,\nstate$ such that sequential states (required for
accessing both the features and the response) are accessible from the
available sample set $\timeDomainSample$.

\subsection{Regression models}\label{sec:regressionModels}

In this work, we consider a wide range of types of models for constructing
regression models $\velocityApproxArg{i}$, $i=1,\ldots,\nstate$.  
We denote a generic regression model
as $\regressionModel$ and its training data as
 \begin{equation} 
	 \trainingData\equiv\{(\responseTrain{i},\featuresTrain{i})\}_{i=1}^{\ntrainingData}
	  \end{equation} 
For all candidate models, we employ validation for model selection; this
process chooses values of the hyperparameters characterizing each model.

\subsubsection{Support vector regression}
\label{section:SVR}
Support vector regression (SVR)~\cite{smola_2004} employs a model
\begin{align}\label{eq:SVRone}
	\regressionModel(\features;\weight) =
	\innerprod{\weight}{\svrMapping(\features)}+b,
\end{align}
where $\svrMapping:\RR{\nfeatures}\rightarrow\SVRfeaturespace$,
$\weight\in\SVRfeaturespace$, and $\SVRfeaturespace$ is a (potentially
unknown) feature space
equipped with inner product $\innerprod{\cdot}{\cdot}$.  SVR aims to compute a
`flat' function (i.e., $\innerprod{\weight}{\weight}$
small) that penalizes prediction errors that exceed a specified threshold
$\epsilon$ (i.e., soft margin loss).  SVR uses slack variables $\boldsymbol{\xi}$ and
$\boldsymbol{\xi}^\star$  to address deviations exceeding  epsilon margin 
$(\epsilon)$ and employs the box constraint $(C)$ to penalize these deviations,
leading to the primal formula~\cite{chang_2011}
\begin{align}\label{eq:SVRopt}
	\begin{split}
		\underset{\weightDummy,b,\boldsymbol{\xi},\boldsymbol{\xi}^\star}{\mathrm{minimize}}\
		&\frac{1}{2}\innerprod{\weightDummy}{\weightDummy}+C\sum_{i=1}^{\ntrainingData}(\xi_i+\xi_i^\star),\\
		\text{subject to}\
		&\responseTrain{i}-\innerprod{\weightDummy}{\svrMapping(\featuresTrain{i})}-b \le \epsilon +
		\xi_i,\quad i=1,\ldots,\ntrainingData,
\\
&\innerprod{\weightDummy}{\svrMapping(\featuresTrain{i})}+b-\responseTrain{i} {}\le \epsilon +
\xi_i^\star,\quad i=1,\ldots,\ntrainingData,
\\
&\boldsymbol{\xi},\,\boldsymbol{\xi}^\star {}\ge \zero,
	\end{split}
\end{align}
whose corresponding dual problem is
\begin{align*}
\begin{split}
	\underset{\boldsymbol{\alpha},\boldsymbol{\alpha}^\star}{\text{minimize}} \
	&\frac{1}{2}(\boldsymbol{\alpha}-\boldsymbol{\alpha}^\star)^T\mathbf{Q}(\boldsymbol{\alpha}-\boldsymbol{\alpha}^\star)
{}+\epsilon \mathbf{1}^T(\boldsymbol{\alpha}+\boldsymbol{\alpha}^\star)
	{}-[\responseTrain{1}\ \cdots\ \responseTrain{\ntrainingData}]^T(\boldsymbol{\alpha}-\boldsymbol{\alpha}^\star),\\
\text{subject to}\ &\mathbf{1}^T(\boldsymbol{\alpha}-\boldsymbol{\alpha}^\star){}=0,
\\
&0 \le \alpha_i,\,\alpha_i^\star{}\le C,\quad i=1,\ldots,\ntrainingData.
\end{split}
\end{align*}
Here, $Q_{ij}\defeq K(\featuresTrain{i},\,\featuresTrain{j})\defeq
\innerprod{\svrMapping(\featuresTrain{i})}{\svrMapping(\featuresTrain{j})}$ and
$\weight=\sum_{i=1}^\ntrainingData(\alpha_i-\alpha_i^\star)\svrMapping(\featuresTrain{i})$.
The resulting model \eqref{eq:SVRone} can be equivalently expressed as
\begin{align}
\regressionModel(\features;\weight) = \sum_{i=1}^{\ntrainingData}
\left(\alpha_i-\alpha_i^\star\right) K(\featuresTrain{i},\features) + b.
\label{eq:SVRtwo}
\end{align}
There exist many kernel functions $K(\featuresTrain{i},\,\featuresTrain{j})$ that
correspond to an inner product in a feature space $\SVRfeaturespace$.
We consider both a Gaussian radial basis function (RBF) kernel
$K(\featuresTrain{i},\,\featuresTrain{j})=\exp(-\gamma\|\featuresTrain{i}-\featuresTrain{j}\|^2)$ and polynomial kernel  $K(\featuresTrain{i},\,\featuresTrain{j})=(1+[\featuresTrain{i}]^T\featuresTrain{j})^q$.
When the RBF kernel is used, we refer to the method as SVRrbf; when using the
polynomial kernel with $q=2$ or $3$, we refer to the method SVR2
and SVR3, respectively. In the numerical experiments, we apply a box constraint
$C = 1$  and select the sensitivity margin $\epsilon$ using a validation set.

\subsubsection{Random forests}
\label{section:RF}
Random forests (RF) \cite{breiman_2001} use decision trees constructed by decomposing feature space (in this case $\RR{\nstate}$) along
canonical directions in a way that sequentially minimizes the mean-squared
prediction error over the training data. The prediction generated by a
decision tree corresponds to the average value of the response over the
training data that reside within the same feature-space region as the
prediction point. 

Decision trees are low bias and high variance; thus, random forests employ a
variance-reduction method, namely bootstrap aggregating (i.e., bagging), to
reduce the prediction variance. Here, $\ntree$ different data sets are
generated by sampling the original training set $\trainingData$ with
replacement. The method then constructs a decision tree from each of these
training sets, yielding $\ntree$ regression functions
$\regressionModelArg{i}$, $i=1,\ldots,\ntree$. The ultimate regression
function corresponds to the average prediction across the ensemble:
\begin{align}
	\regressionModel(\features) = \frac{1}{\ntree}\sum_{1=1}^\ntree \regressionModelArg{i}(\features).
	\label{eq:regressiontree}
\end{align}
To decorrelate the decision trees and further reduce prediction
variance, random forests introduce an additional source of randomness. When training
each tree, a random subset of $\nfeaturesSplit\ll \nfeatures$ features is
considered when performing the feature-space split at each node in the tree.

The hyper-parameter we consider for this approach corresponds to the number of trees in
the ensemble $\ntree$.  We set $\nfeaturesSplit = \nfeatures/3$  for splitting during training.

\subsubsection{Boosted decision trees}
\label{sec:boosting}
Boosted decision trees \cite{Hastie2009} combines weak learners (decision trees) into a single strong learner in an iterative fashion. The algorithm adds one tree at each stage.
Let $\LSBoosting_m$ be the aggregate model at stage $m$, where
$m=1,\ldots,\ntree$. The gradient boosting algorithm improves upon
$\LSBoosting_m$ by constructing a new model that adds a weak learner
$\WeekLearner_{m}$ to yield a more accurate model
$\LSBoosting_{m+1}(\features)=\LSBoosting_{m}(\features)+\alpha_m \WeekLearner_{m}(\features)$, where $\alpha_m$ is a constant.

The numerical experiments employ LSBoost, where the goal is to minimize the
mean squared error 
$
\sum_{i=1}^\ntrainingData|\regressionModel(\featuresTrain{i})-\responseTrain{i}|^2
$.  Decision stumps are applied as weak learners $ \WeekLearner_{m}(\features)$. Initially, $\LSBoosting_{0} = \WeekLearner_0(\features)$.
At iteration $m$ (for $m = 1, \cdots, \ntree$), compute the residual $r_i = \responseTrain{i} - \LSBoosting_{m -1}(\featuresTrain{i})$ ($i=1, \ldots, \ntrainingData$) and solve the optimization problem
$$(\alpha_m, \WeekLearner_{m}) = {\rm{arg \  min}}_{\alpha,  \WeekLearner} \sum_{i=1}^\ntrainingData |r_i - \alpha \WeekLearner(\featuresTrain{i})|^2.$$ The final result is given by $\regressionModel(\features) = \LSBoosting_{\ntree}(\features)$.

The hyper-parameter we consider for this approach corresponds to the number of
trees $\ntree$ in the ensemble.

\subsubsection{$k$-nearest neighbors}
\label{section:knn}

The $k$-nearest neighbors ($k$-NN) method \cite{altman1992introduction}
produces predictions corresponding to a weighted average of the responses
associated with the $k$-nearest training points in feature space, i.e.,
\begin{align*}
	\regressionModel(\features) = \sum_{i\in\nearestSetArg{k}}
	\knnWeight(\featuresTrain{i},\features)
	\featuresTrain{i}.
\end{align*}
Here, $\nearestSetArg{k}\subseteq\{1,\ldots,\ntrainingData\}$ with
$\card{\nearestSetArg{k}}=k(\leq\ntrainingData)$ satisfies
$\|\features-\featuresTrain{i}\|_2\leq 
\|\features-\featuresTrain{j}\|_2$ for all $i\in\nearestSetArg{k}$,
$j\in\{1,\ldots,\ntrainingData\}\setminus \nearestSetArg{k}$.
In the numerical experiments, we consider only uniform
weights
$\knnWeight(\featuresTrain{j},\features)\defeq\frac{1}{k}$.

The hyper-parameter we consider for this method corresponds to the number of
nearest neighbors $k$. 

\subsubsection{Vectorial kernel orthogonal greedy algorithm (VKOGA)}
\label{sec:VKOGA}
In support vector machines, kernel-based interpolation is applied to
scalar-valued functions (as in Eq.~\refeq{eq:SVRtwo}). If each individual
regression model $\velocityApproxArg{i}$, $i=1,\ldots,\nstate$
requires $\ntrainingData$  sets of kernels to construct its approximation, 
 the resulting regression model for the vector of responses $\velocityApprox$
 would require  $\nstate \ntrainingData$ independent kernels. This is
 expensive for both training and evaluation when $\nstate \gg
 1$. To reduce the overall number of kernels, one can impose the 
restriction that a common subspace be used for every component of the
vector-valued response. In particular, the vectorial kernel orthogonal greedy
algorithm (VKOGA) \cite{wirtz2013vectorial,wirtz2015surrogate} yields a regression
model of the form
\begin{align}
\velocityApprox(\features) = \sum_{i=1}^{\nVKOGA}
\VKOGAbasis_i K(\features_{i},\features).
\label{eq:VKOGAExp}
\end{align}
where $K:\RR{\nstate\times\nstate}\rightarrow\RR{}$ denotes a kernel function,
$\featuresTrain{i}\in\RR{\nstate}$, $i=1,\ldots,\nVKOGA$ denotes the kernel centers and $\VKOGAbasis_i\in\RR{\nstate}$, $i=1,\ldots,\nVKOGA$ denote vector-valued
basis functions. 

VKOGA first computes kernel functions $K(\featuresTrain{i}, \cdot)$ by a greedy algorithm. The greedy algorithm determines kernel centers from $\Omega =\{\featuresTrain{1}, \ldots, \featuresTrain{\ntrainingData}\}$. Initially, let $\Omega_0 = \emptyset$. 
At stage $m$, choose 
$$\features_m := \mathop {\rm arg max}_{\features \in \Omega \backslash \Omega_{m-1}} | \langle \velocity, \boldsymbol{\phi}_x^{m-1}  \rangle |,$$
 where   $\boldsymbol{\phi}_x^{m-1}$ denotes the orthogonal remainer of $K(\features, \cdot)$ with respect to the reproducing kernel Hilbert space spanned by $\{K(\features_{1}, \cdot), \ldots, K(\features_{m-1}, \cdot) \}$. Then $\Omega_m = \Omega_{m-1} \cup \{\features_m\}$.
 After
the kernel centers have been computed, the basis functions $\VKOGAbasis_i$,
$i=1,\ldots,\nVKOGA$ are determined by a least-squares approximation to best
fit training data. 

In the numerical experiments, we apply Gaussian RBF kernel
$K(\featuresTrain{i},\,\featuresTrain{j})=\exp(-\gamma\|\featuresTrain{i}-\featuresTrain{j}\|^2)$.
 The hyper-parameter we consider for this method corresponds
to the number of kernel functions $\nVKOGA$.

\subsubsection{Sparse identification of nonlinear dynamics (SINDy)}

The sparse identification of nonlinear dynamics (SINDy) method
\cite{brunton2016discovering} is an application of 
linear regression to learning dynamics. Although SINDy was originally devised
for continuous-time dynamics, it can be easily extended to the discrete-time
case. In this context, it constructs a model
\begin{equation}\label{eq:SINDy} 
	\regressionModel(\features;\weight) =
	\sum_{i=1}^\nSINDyterm\basisFunctionArg{i}(\features)\weightArg{i} 
\end{equation} 
where $\basisFunctionArg{i}:\RR{\nstate}\rightarrow\RR{}$,
$i=1,\ldots,\nSINDyterm$ denote a ``library'' of prescribed basis functions.
The least absolute shrinkage and selection operator
(LASSO)~\cite{tibshirani1996regression} is used to determine a sparse set of
coefficients $(\weightArg{1},\ldots,\weightArg{\nSINDyterm})$.

Hyperparameters in this approach consist of the selected basis functions.  For
simplicity, we consider only linear and quadratic functions in the library,
i.e., $\basisFunctionArg{i} \in  \{ \stateArgs{1}{}, \ldots,
\stateArgs{\nstate}{},  {\stateArgs{1}{}} {\stateArgs{1}{}}, \stateArgs{1}{}
\stateArgs{2}{}, \ldots, \stateArgs{\nstate}{} \stateArgs{\nstate}{} \}$.

\subsubsection{Dynamic mode decomposition (DMD)}

Dynamic mode decomposition (DMD) \cite{schmid2010dynamic}
computes the approximated velocity $\velocityApprox$ as the linear operator
\begin{equation} 
	\velocityApprox:\state\mapsto 
	[\stateArg{\timeArgRegression{2}}\ \cdots\
	\stateArg{\timeArgRegression{\ntrainingData+1}}]
		[\stateArg{\timeArgRegression{1}}\ \cdots\
	\stateArg{\timeArgRegression{\ntrainingData}}
		]^+	
	\state
\end{equation} 
where the superscript $+$ denotes the Moore--Penrose pseudoiverse.
Critically, note that this approach ensures that $\stateArg{n+1} =
\velocityApprox(\stateArg{n})$, $n\in\timeDomainTrainRegression$ if
$\nstate\geq \ntrainingData$.

\subsection{Boundedness}\label{sec:boundedness}

This section provides analysis related to the stability of the approximated
discrete-time dynamical system \eqref{eq:discTimeApprox} when the different
regression methods described in Section \ref{sec:regressionModels} are
employed to generate the approximated velocity $\velocityApprox$.

First, we note that a function $\velocitySymb: \xspace \to \RR{\nstate}$ is
\textit{bounded} if the set of its values is bounded. In other words, there
exists a real number $\bignumber$ such that $\|\velocitySymb(\stateSymb)\| \le
\bignumber$ for all $\stateSymb \in \xspace$. An important special case is a
bounded sequence, where $\xspace$ is taken to be the set $\mathbb{N}$ of
natural numbers. Thus a sequence $$(\stateArg{0}, \stateArg{1}, \stateArg{2}, \ldots)$$
is bounded if there exists a real number $\bignumber$ such that $\|
\state(\timeSymb) \| \le \bignumber$ for every natural number $\timeSymb$.
With $\stateApproxArg{n+1} = \velocityApprox(\stateApproxArg{n})$, boundedness of
$\velocityApprox: \RR{\nstate} \to \RR{\nstate}$ yields boundedness of the sequence
$(\stateApproxArg{0}, \stateApproxArg{1}, \stateApproxArg{2}, \ldots)$.

\begin{lemma}
	Let $\idset$ be an index set. With $\rho_\id: \RR{\nstate} \to
	\RR{}$, $\{\rho_\id\}_{\id \in \idset}$ is a set of scalar
	basis functions. Suppose $\velocityApprox \in {\rm{span}}\left( \{
		\rho_\id\}_{\id \in \idset} \right)$, i.e.,
	there
	exist $\VKOGAbasis_i\in\RR{}$, $i=1,\ldots, n$  (with $n$ finite)
	such that
$$\velocityApprox(\state) = \sum_{i=1}^n  \VKOGAbasis_i \rho_\id(\state).$$ 
	If each basis function
	$\rho_\id$ is bounded, then $\velocityApprox$ is bounded.
\end{lemma}
\begin{proof}
	Let $\maxcoeff = \max_{i=1}^n  \|\VKOGAbasis_i \| $.
Since $\rho_\id(\cdot)$ is bounded, $\|\rho_\id(\features)\|\le \bignumber_i$
	for all
$\state \in \RR{\nstate}$. Let $\bignumber = \max_{i=1}^n \bignumber_i$. Thus,
$$  \|\velocityApprox(\state)\| = \left \| \sum_{i=1}^n  \VKOGAbasis_i \rho_\id(\state) \right \|   \le  \sum_{i=1}^n  \left \|  \VKOGAbasis_i \rho_\id(\state) \right\| \le \sum_{i=1}^n \| \VKOGAbasis_i \| \cdot \|\rho_\id(\state) \| \le n \maxcoeff  \bignumber.   $$
The last expression provides an upper bound for $\|\velocityApprox(\state)\|$.
\end{proof}

\begin{corollary}
	If the approximated velocity $\velocityApprox$ is constructed using SVRrbf
	(as described in Section \ref{section:SVR}) or VKOGA (as described in
	Section \ref{sec:VKOGA}), then the approximated velocity is bounded.
\end{corollary}
\begin{proof}
Since a Gaussian kernal satisfies
$K(\featuresTrain{i},\, \state)=\exp(-\gamma\|\featuresTrain{i}- \state \|^2) \le 1$ for all $\state \in \RR{\nstate}$,
	it is bounded.
		By the previous lemma, with
$\rho_\id(\state) = K(\featuresTrain{i},\, \state)$, 
	the mapping $\velocityApprox$ is also bounded. 
\end{proof}

With the previous lemma, as long as the basis function is bounded, then the
function $\velocityApprox$ is also bounded. Thus, $k$-NN, Random forest, and
Boosting can also yield a bounded sequence  $(\state(0), \state(1), \state(2),
\ldots)$. In contrast, basis functions associated with DMD are linear, basis
functions associated with 
SINDy are polynomial, and basis functions associated with SVR with polynomial
kernels are
$K(\featuresTrain{i},\,\featuresTrain{j})=(1+[\featuresTrain{i}]^T\featuresTrain{j})^q$,
none of which are bounded.

\section{Numerical experiments}\label{sec:experiments}

To assess the proposed methodology, we consider a large-scale CFD problem with
a high-order discontinuous Galerkin discretization.

\subsection{Test case and discretization}

Our chosen test case is the flow over an extruded half-cylinder at Reynolds
number $\Reynolds = 350$ and an effectively incompressible Mach number of $M =
0.2$.  This case contains several complex flow features, including separated
shear layers, turbulent transition, and a fully turbulent wake.  When compared
with the circular cylinder, the dynamics of the half-cylinder are noteworthy
because the point of separation is essentially fixed.
Flow over such a cylinder has been the subject of several numerical and
experimental studies \cite{nakamura1996vortex, kumarasamy1997computation,
santa2005characterization}.

In the present study, the cylinder is taken to have a diameter $D$ along its
major axis.  The domain is taken to be $[-9D, 25D]$, $[-9D, 9D]$, and $[0, \pi
D]$ in the stream-, cross-, and span-wise directions, respectively.  The
cylinder is positioned such that the back surface is centered at $(0,0,0)$.
The stream-wise and cross-wise dimensions are comparable to those used in
studies of the circular cylinder \cite{parnaudeau2008experimental,
witherden2015heterogeneous}.  The domain is periodic in the span-wise
direction.  We remark here that, at this Reynolds number, the dynamics are
strongly influenced by the span-wise extent of the domain.  Indeed, for
extrusions less than ${\sim}\pi D/2$, the wake is coherent.  We apply a no-slip
adiabatic wall boundary condition at the surface of the cylinder
and Riemann invariant boundary conditions at the far-field.
To non-dimensionalize the system, we take the cylinder diameter to be $D = 1$,
the free-stream density to be one, and the free-stream pressure to be one.
The free-stream velocity is thus $v_{\infty} = \sqrt{\gamma} M \simeq 0.234$,
where $\gamma = 1.4$ is the ratio of specific heats.

We mesh the domain using quadratically curved hexahedral elements.  In total
the mesh has $\nelements=40\,584$ elements.  An overview of the mesh can be
seen in Figure \ref{fig:mesh}.
\begin{figure}
 \centering
 \includegraphics[width=8cm]{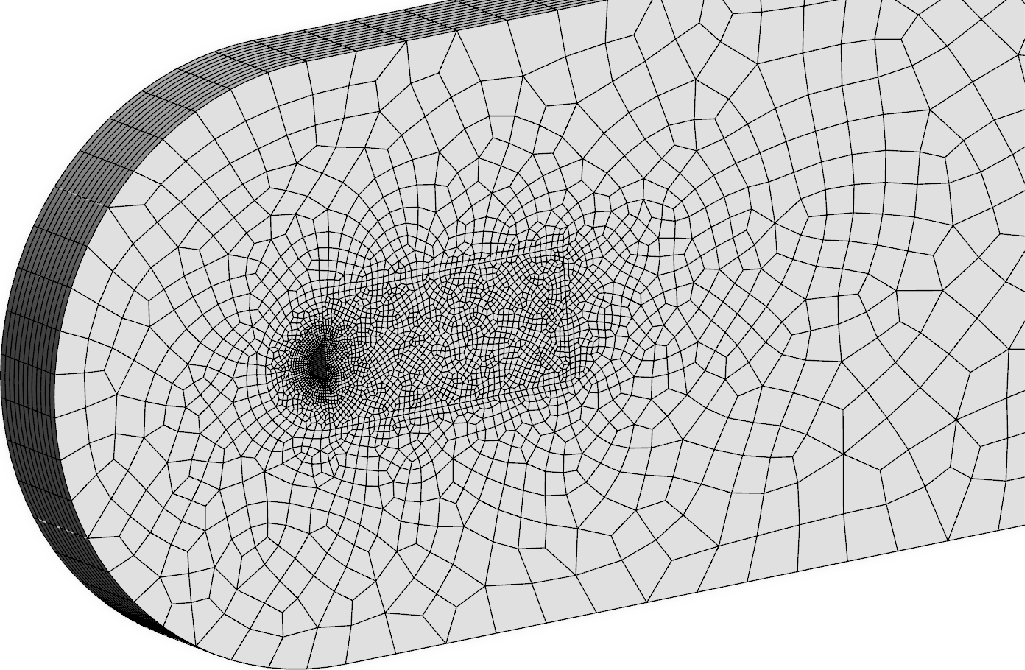}
 \caption{\label{fig:mesh}The half-cylinder mesh characterized by $40\,584$ elements.}
\end{figure}
We solve the compressible Navier--Stokes equations on this mesh with viscosity
fixed such that the Reynolds number based on the diameter of the cylinder is
$\Reynolds = 350$.  For the spatial discretization, we employ a high-order
nodal DG scheme with third-order solution polynomials and Gauss--Legendre
solution points.  Due to the third-order solution polynomials, the
local state within each element contains the density, $x$-momentum, $y$-momentum, $z$-momentum, and energy at
64 points, leading to $\nstateLocalVector=5 \times 64 = 320$.
Thus, the resulting dimension of the global state is
$\nstateGlobalVector{\approx}13 \times 10^6$.  This implies that saving a given
solution snapshot to disk consumes ${\sim}{100}\text{MiB}$ of storage
using double precision arithmetic. 
We calculate inviscid fluxes between
elements using a Rusanov-type Riemann solver, and we take the Prandtl number
 to be $\Prandtl = 0.72$.  For time integration,  we employ the
explicit five-stage fourth-order RK45[2R+] scheme of Carpenter and Kennedy
\cite{kennedy1999low}, which manages the local temporal error by utilizing a
PI type controller to adaptively modify the time step.

We start the simulation cold with a constant initial condition corresponding to the free-stream at $t = 0$ and terminate the simulation at $t = 600$.  This
corresponds to approximately $71$ stream-wise passes over the half-cylinder.
Our time grid of interest comprises $t \in \{170.05 + 0.05i\}_{i=1}^\ntime$
with $\ntime=8\,600$ total time instances. Thus, the desired sequence of states
$
	\{\stateGlobalVectorArg{n}\}_{n=0}^\ntime\subset\RR{\nstateGlobalVector}
	$
	corresponds to the CFD solution at these points in time.

\subsection{Dimensionality Reduction}

As described in Section \ref{sec:dimRed}, the first stage of our proposed
methodology is dimensionality reduction, which proceeds in two steps: local
compression using autoencoders, and global compression using PCA.

\subsubsection{Local compression using autoencoders}

\begin{sidewaysfigure}[p]
\centering
    \includestandalone[mode=image, width=0.98\linewidth]{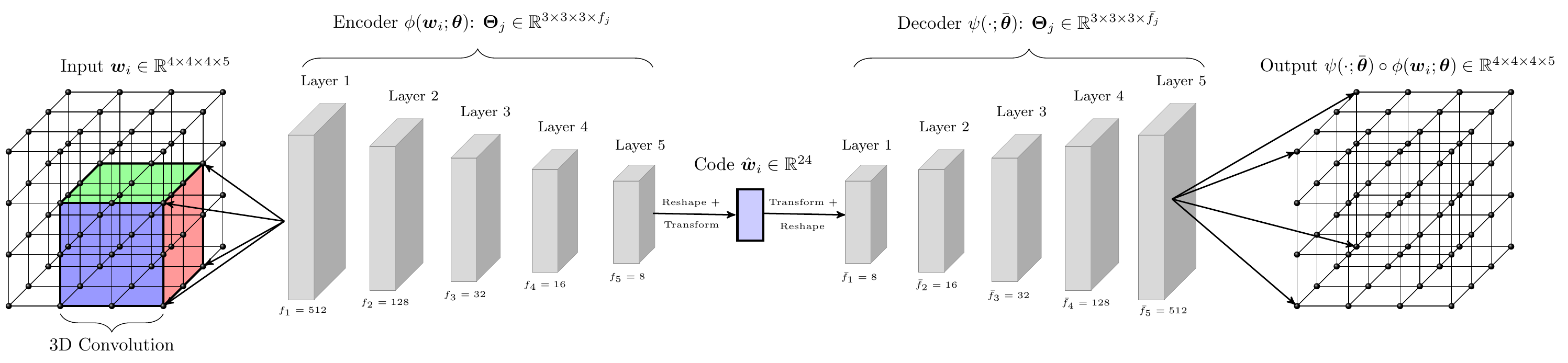}
    \caption{Illustration of the network architecture. In the encoder, the input is transformed by a series of convolutional layers with a decreasing number of filters at each layer. The code is obtained by reshaping the output of the final encoder layer into a vector and performing an affine transformation. The decoder outputs a reconstruction of the original input $w_i$ by performing the inverse of all operations performed by the encoder.}
    \label{fig:network}
\end{sidewaysfigure}

Because convolutional neural networks were originally devised for image data, they
require the local degrees of freedom (i.e., the five conserved variables) to be
equispaced within each element.  However, within the simulation, these
variables are defined on the tensor product of a 4-node
Gauss--Legendre quadrature rule in order to minimize aliasing errors when
projecting the non-linear flux into the polynomial space of the solution.
Thus, before applying the autoencoder, we transform the solution by
interpolating it from these Gauss--Legendre points to a $4\times 4\times
4$
grid of equispaced points. The dimension of the network inputs is therefore $4 \times 4 \times 4 \times 5$, where the first three dimensions correspond to
spatial dimensions and the last dimension corresponds to the flow quantities.
We do not account for the fact that the mesh elements can vary in size, and
hence the spacing between points is not uniform across all training examples.
Since solution values at both sets of points define the same polynomial
interpolant, the solution inside of each element is unchanged by this
transformation.

For the autoencoder, we apply convolutional layers with ResNet skip
connections, as they have previously been shown to achieve state-of-the-art
performance on image classification and related tasks~\cite{he2016resnet}.
Between convolutional layers, we apply batch normalization~\cite{ioffe2015batch} and ReLU activations~\cite{nair2010rectified}.
The encoder network has $\nlayersEncoder= 5$, where each layer contains parameters associated with a set of three-dimensional convolutional filters $\weightNNEncoderarg{i} \in \R^{3 \times 3 \times 3 \times f_i}$, and $f_i$ denotes the number of filters contained in layer $i$.  There are \num{512},
\num{128}, \num{32}, \num{16}, and \num{8} filters in the five layers of the encoder. 
\reviewer{This architecture is representative of a standard autoencoder~\cite{hinton2006reducing}, where the dimensionality of the transformed input is initially large to allow for adequate feature extraction, and is subsequently decreased gradually through the remaining hidden layers of the encoder.
While we found that this autoencoder architecture allowed for a high level of compression and reconstruction accuracy, it is possible that other autoencoder architectures may achieve satisfactory performance.}
In all layers, the size of the final output dimension is equal to $f_i$.
In the first four layers of the encoder, the convolutions are performed with a stride of one, meaning that the first three dimensions of the layer output $\nnFuncEncoderArg{i}(\mathbf{y}_{i-1};\weightNNEncoderarg{i})$ are equal in size to the first three dimensions of the layer input $\mathbf{y}_{i-1}$.  In the final layer of the encoder, the convolutions are performed with a stride of two, meaning that the first three dimensions of the layer output $\nnFuncEncoderArg{5}(\mathbf{y}_{4};\weightNNEncoderarg{5})$ are half the size of the first three dimensions of the layer input $\mathbf{y}_{4}$.
Given all of these transformations, we have $\nNeuronEncoder{0} = 320$,
$\nNeuronEncoder{1} = 32\,768$, $\nNeuronEncoder{2} = 8\,192$,
$\nNeuronEncoder{3} = 2\,048$, $\nNeuronEncoder{4} = 1\,024$, and $\nNeuronEncoder{5} = 64$.

The output of the final convolutional layer is reshaped into a vector and
subsequently mapped to encoding $\stateLocalVectorRed{i}$ with an affine
transformation.  We set the reduced dimension of the local state vector (i.e.,
the code dimensionality) to $\nstateLocalVectorRed = 24$.  We then use an
affine transformation and reshaping operation to create an input to the
decoder that is the same size as the output of the encoder.  We construct the
decoder network to invert all operations performed by the encoder in order to
obtain reconstructed solutions; this implies that $\nlayersDecoder=5$,
$\nNeuronDecoder{0} = 64$, $\nNeuronDecoder{1} = 1\,024$, $\nNeuronDecoder{2}
= 2\,048$, $\nNeuronDecoder{3} = 8\,192$, $\nNeuronDecoder{4} = 32\,768$, and
$\nNeuronDecoder{5} = 320$. Figure \ref{fig:network} provides an illustration
of the network architecture.

The training data for the autoencoder consists of the 90 randomly selected
values of the set of time instances $\{86i\}_{i=1}^{100}$ such that
$\card{\timeDomainTrainAuto}=90$; the remaining 10 points of the set comprise
validation data that is used to check for overfitting.
To solve the minimization problem \eqref{eq:autoencoderOpt}, we employ
stochastic gradient descent with the Adam optimizer~\cite{kingma2014adam}.
We apply $L_2$ regularization to
the network weights
\begin{equation}
\regularization(\weightEncoder,\weightDecoder) = \lambda \left(\|\weightEncoder\|_2^2 + \|\weightDecoder\|_2^2   \right),
\end{equation}
where $\lambda$ is a hyperparameter that controls the level of regularization. 
When an increase in the validation error is observed across training epochs,
the learning rate is cut in half.  Training is terminated once the loss on the
validation set fails to decrease across three training epochs, which typically
occurs after the learning rate has been decayed many times.

\subsubsection{Global compression using PCA}\label{sec:PCAexperiments}

We apply PCA for global compression as described in Section \ref{sec:PCA}
using training time instances\footnote{Note that PCA is robust with respect to
missing data points; thus, we expect the principal components computed from
this set to be close to those computed on a smaller training set.} of
$\timeDomainTrainPCA= \timeDomainSample$.
We decompose the global reduced state vector $\stateGlobalVectorRed$ into
$\nPCA=38$ components, each of which is characterized by 24 codes and $1\,068$
cells such that $\stateGlobalVectorRedArg{i}\in\RR{25\,632}$,
$i=1,\ldots,38$. We truncate the component principal components such that
$\nstateArg{i}=500$, $i=1,\ldots,38$.

Figure \ref{fig:svals} reports the decay of the first 500 (of $8\,600$) singular
values (i.e., the first 500 diagonal elements of $\Sing$ in
Eq.~\eqref{eq:SVDGlobal}). In this case, the first 100 principal components
capture approximately 95\% of the statistical energy (as measured by the sum
of squares of all $8\,600$ singular values) in the data, while the first 500
principal components capture approximately 98\% of the statistical energy. We first
set $\nstate=100$.
\begin{figure}[t]
\centering
    \includegraphics[width=0.5\linewidth]{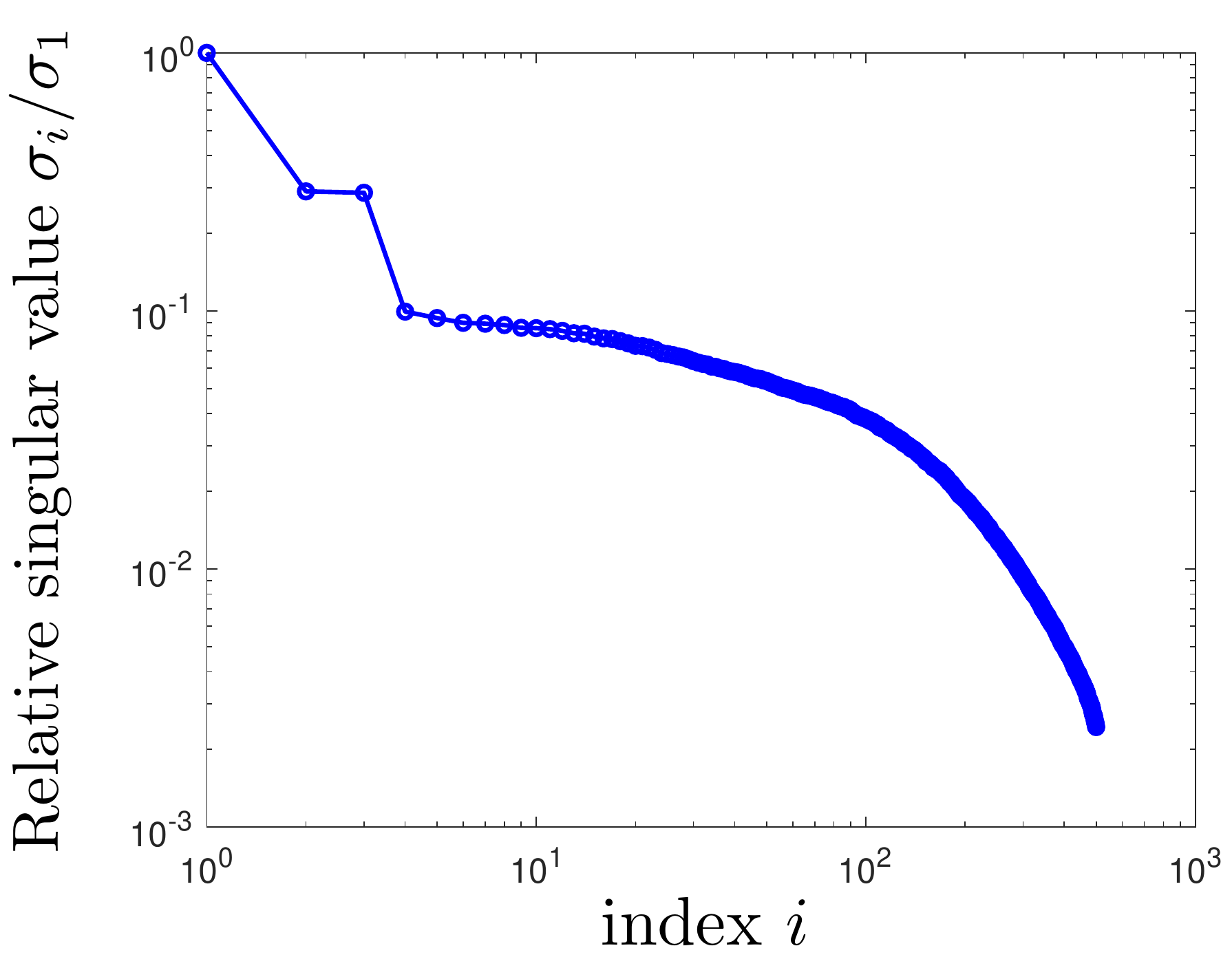}
			\caption{Decay of the \reviewer{singular values}. Numerical computation shows that a basis dimension of
			$\nstate=100$ captures approximately 95\% of the statistical energy
			in the data, and  a basis dimension of
			$\nstate=500$ captures approximately 98\% of the statistical energy.}
  \label{fig:svals}
\end{figure} 

\subsection{Dynamics learning}

We now apply the approach described in
Section \ref{sec:regressionDiscrete} for dynamics learning. We employ a maximum set of training time
instances $\timeDomainTrainRegression$ with
$\card{\timeDomainTrainRegression}= 6 \,450$. 
Before applying regression, we standardize the features by applying an affine
transformation (i.e., we subtract the sample mean and divide by the sample
standard deviation).

\subsubsection{Hyperparameter selection using
validation}\label{sec:hyperparamSelect}

We now present validation results for the regression methods with
tunable hyperparameters, i.e., SVR2, SVR3, SVRrbf, random forest, boosting, $k$-NN, and VKOGA.
We set the validation set to be the $2\,150$ time instances not
included in the maximum training set.

Figure \ref{fig:validation} reports these results, wherein the models are
trained using the maximum training set  $\timeDomainTrainRegression$. 
Here, the relative mean-squared error (MSE) over a set of time indices
$\timeDomainGen\subseteq\{1,\ldots,\ntime\}$ is defined
as
\begin{equation}\label{eq:relativeMSE}
\frac{\sum_{n\in\timeDomainGen}\|
\velocityApprox(\stateArg{n})- \stateArg{n+1} \|_2}{
\sum_{n\in\timeDomainGen}\|
\stateArg{n+1}\|_2
}.
\end{equation}
Based on
these results, we select the following hyperparameters: $\epsilon = 10^{-3}$
for SVR2, SVR3, and SVRrbf; $\ntree=30$ for random forest; $\ntree = 100$ for
boosted decision trees; $k=6$ for $k$-nearest neighbors; and 200 kernel
functions for VKOGA.

\begin{figure}[htbp] 
    \begin{subfigure}[b]{0.33 \linewidth}
  \centering
    \includegraphics[width=\linewidth]{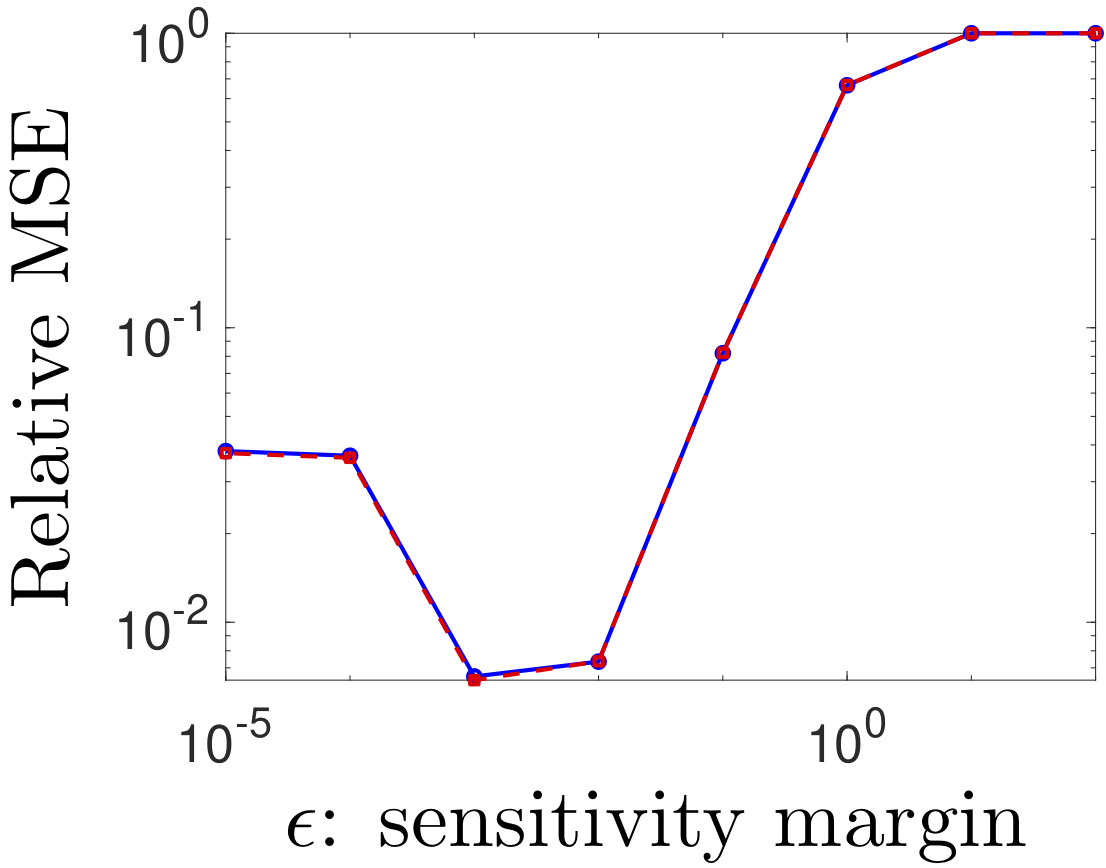}
			\caption{SVR2}
  \end{subfigure}
    \begin{subfigure}[b]{0.33 \linewidth}
  \centering
    \includegraphics[width=\linewidth]{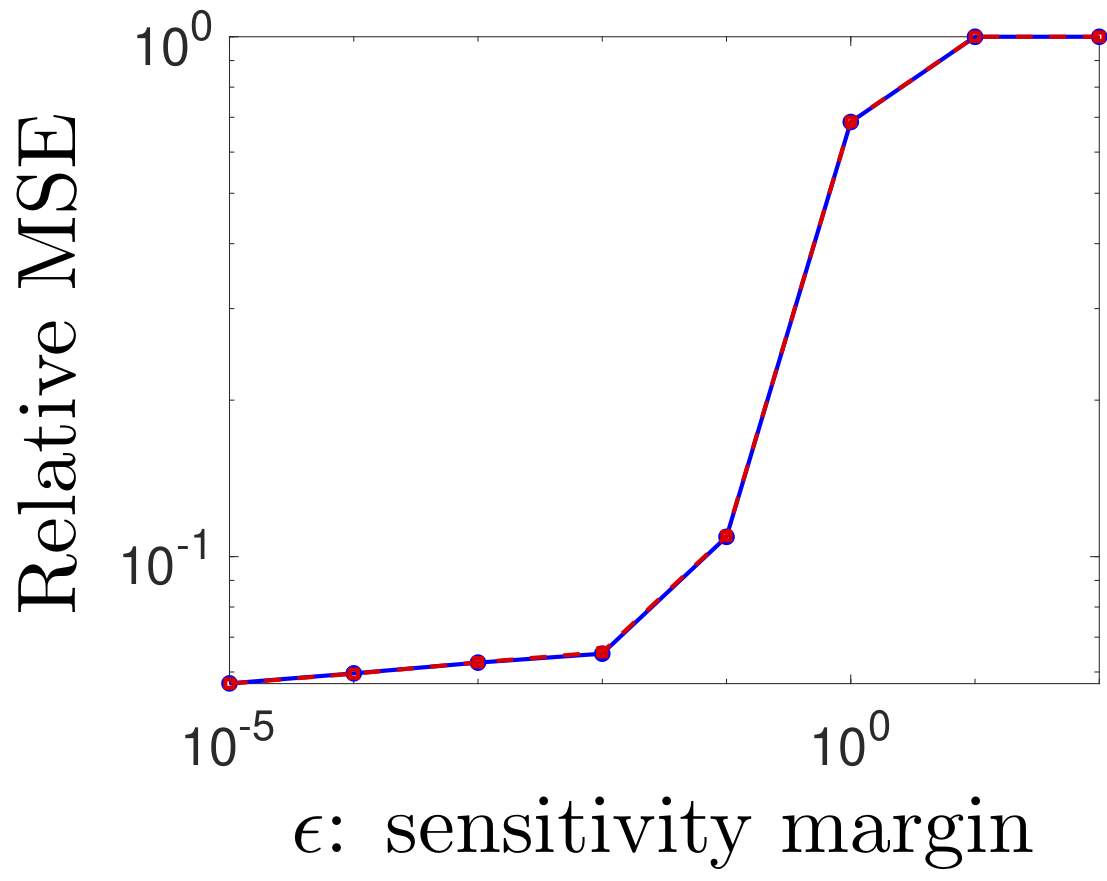}
			\caption{SVR3}
  \end{subfigure}
    \begin{subfigure}[b]{0.33 \linewidth}
  \centering
    \includegraphics[width=\linewidth]{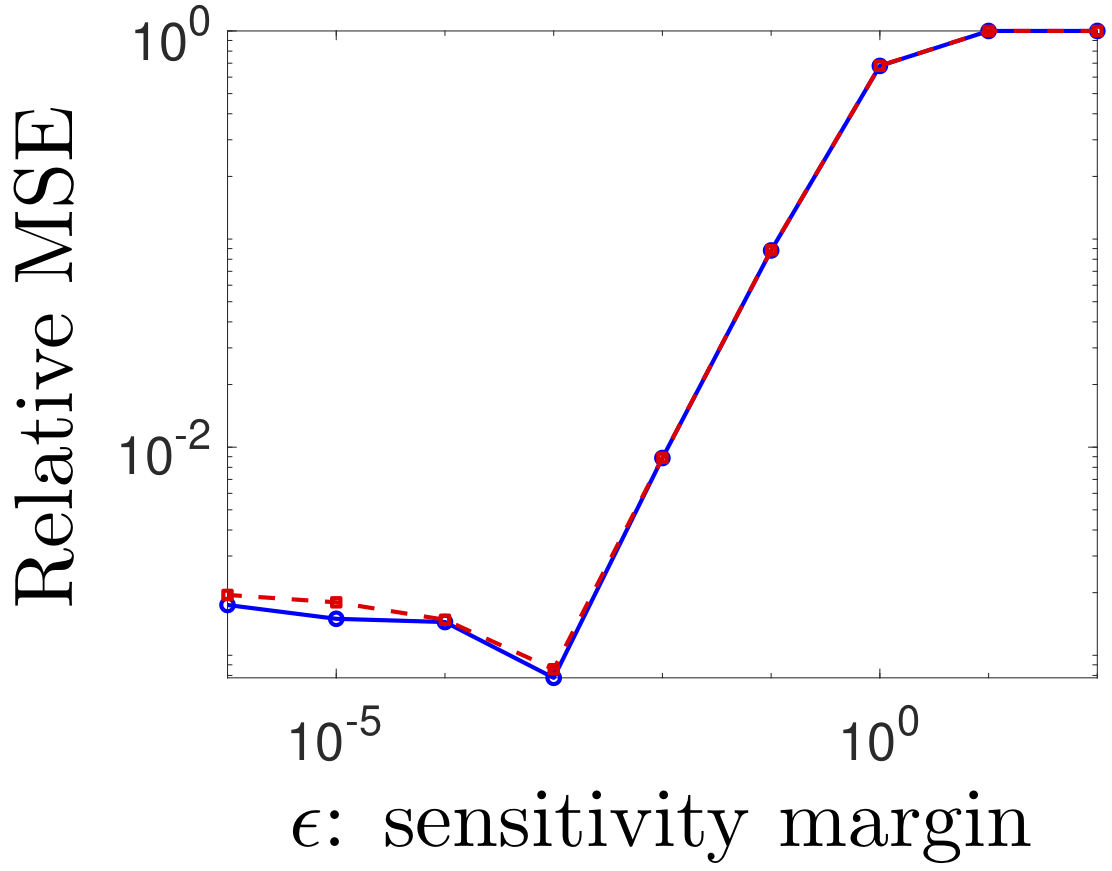}
			\caption{SVRrbf}
  \end{subfigure}
    \begin{subfigure}[b]{0.33 \linewidth}
  \centering
    \includegraphics[width=\linewidth]{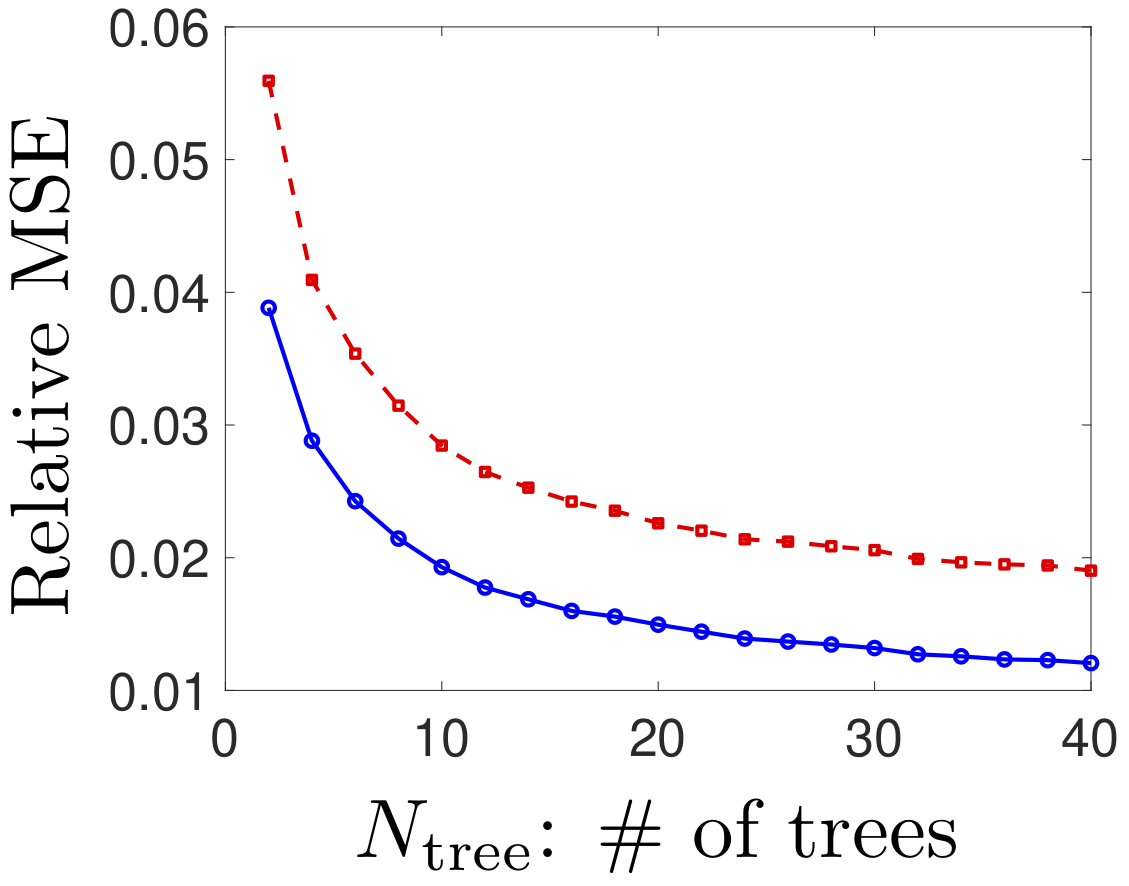}
			\caption{Random Forest}
  \end{subfigure}
    \begin{subfigure}[b]{0.33 \linewidth}
  \centering
    \includegraphics[width=\linewidth]{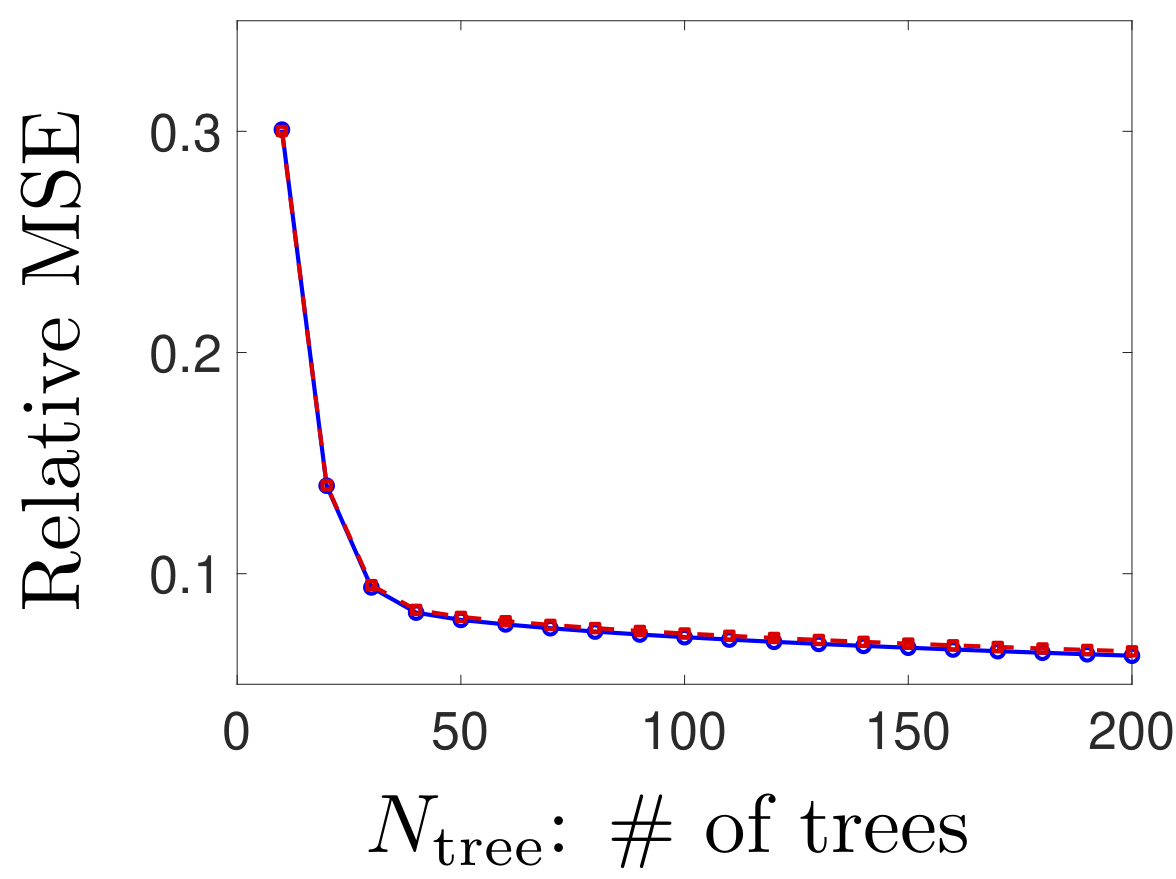}
			\caption{Boosted decision trees}
  \end{subfigure}
    \begin{subfigure}[b]{0.33 \linewidth}
  \centering
    \includegraphics[width=\linewidth]{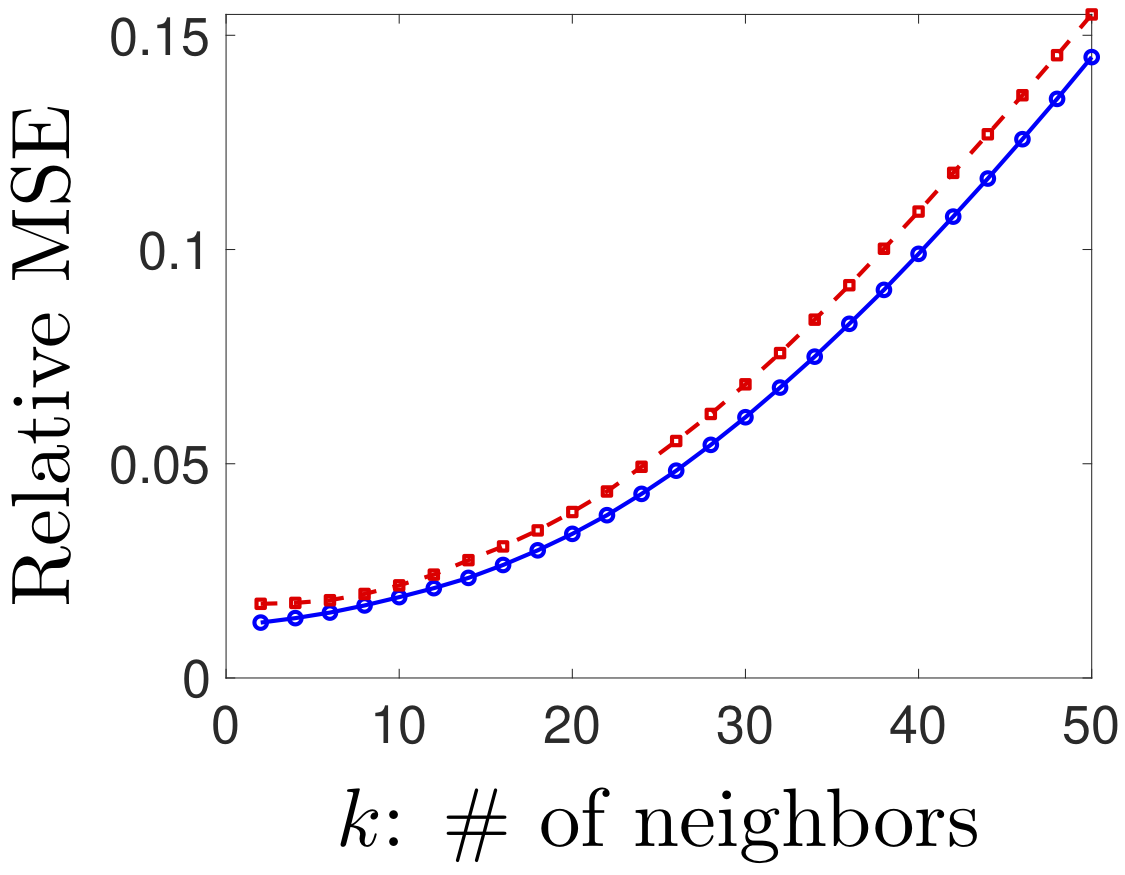}
			\caption{$k$-nearest neighbors}
  \end{subfigure}
    \begin{subfigure}[b]{0.33 \linewidth}
  \centering
    \includegraphics[width=\linewidth]{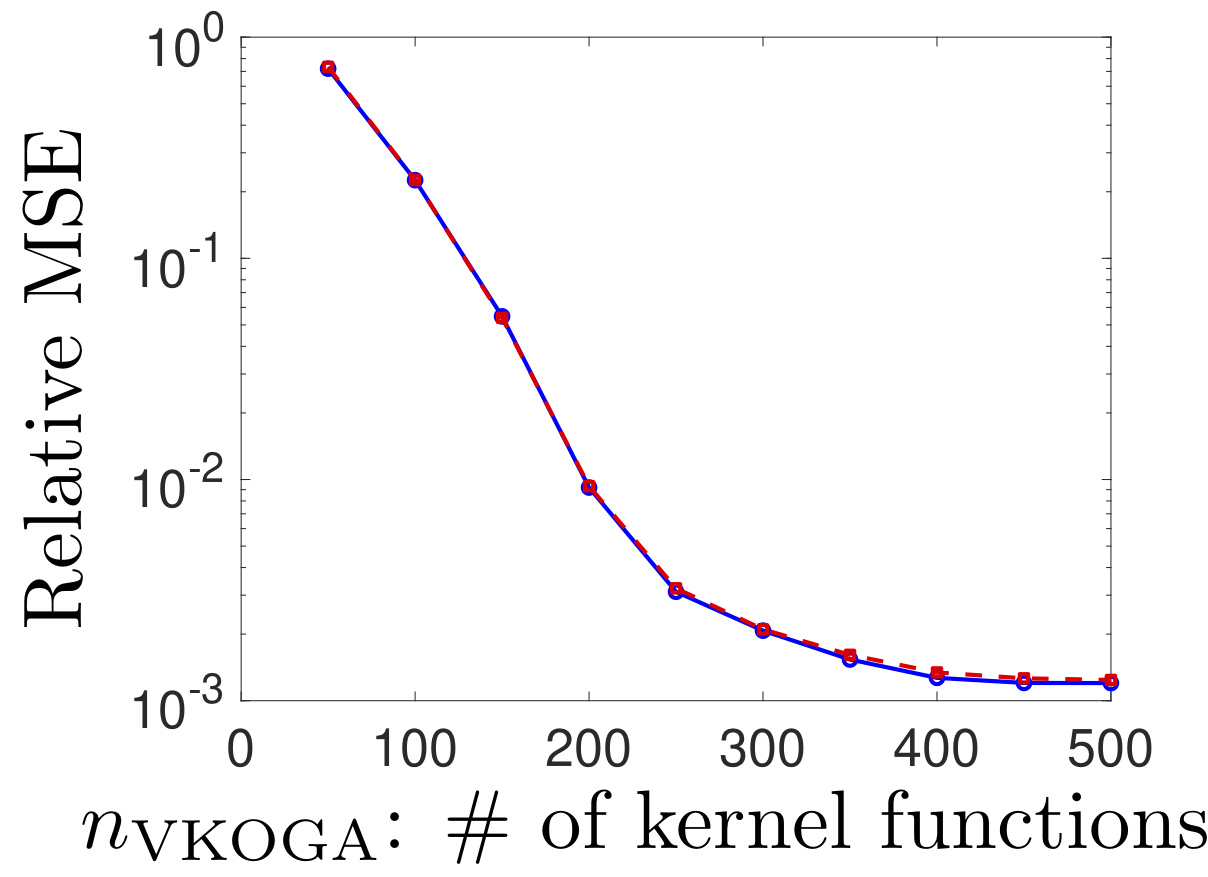}
			\caption{VKOGA}
  \end{subfigure}
\caption{Hyperparameter selection for different regression methods.  Reported
	errors correspond to the relative MSE values on the training set (blue) and
	the validation set (red).  The training set corresponds to the maximum
	training set $\timeDomainTrainRegression$, and the validation set
	corresponds to the remaining $2\,150$ time instances.} 
\label{fig:validation} 
\end{figure} 

\subsubsection{Training and testing error} 

We now analyze the training and testing errors of the various regression
methods, where the hyperparameters have been fixed according to the values
selected using validation in Section \ref{sec:hyperparamSelect}.
Figure
\ref{fig:trainingTestingError} reports these results. Here, three different
training sets are considered to construct the regression models: one with the maximum 
training set $\timeDomainTrainRegression$ with
$\card{\timeDomainTrainRegression} = 6\,450$ (yellow bars), one with 75\% of
the maximum training
data selected at random (red bars), and one with 50\% of the maximum
training data selected at random (blue bars).
The reported errors correspond to the relative MSE as defined in
Eq.~\eqref{eq:relativeMSE}.

\begin{figure}[H]
	\begin{center}
    \begin{subfigure}[b]{0.48 \linewidth}
  \centering
		\includegraphics[width=\textwidth]{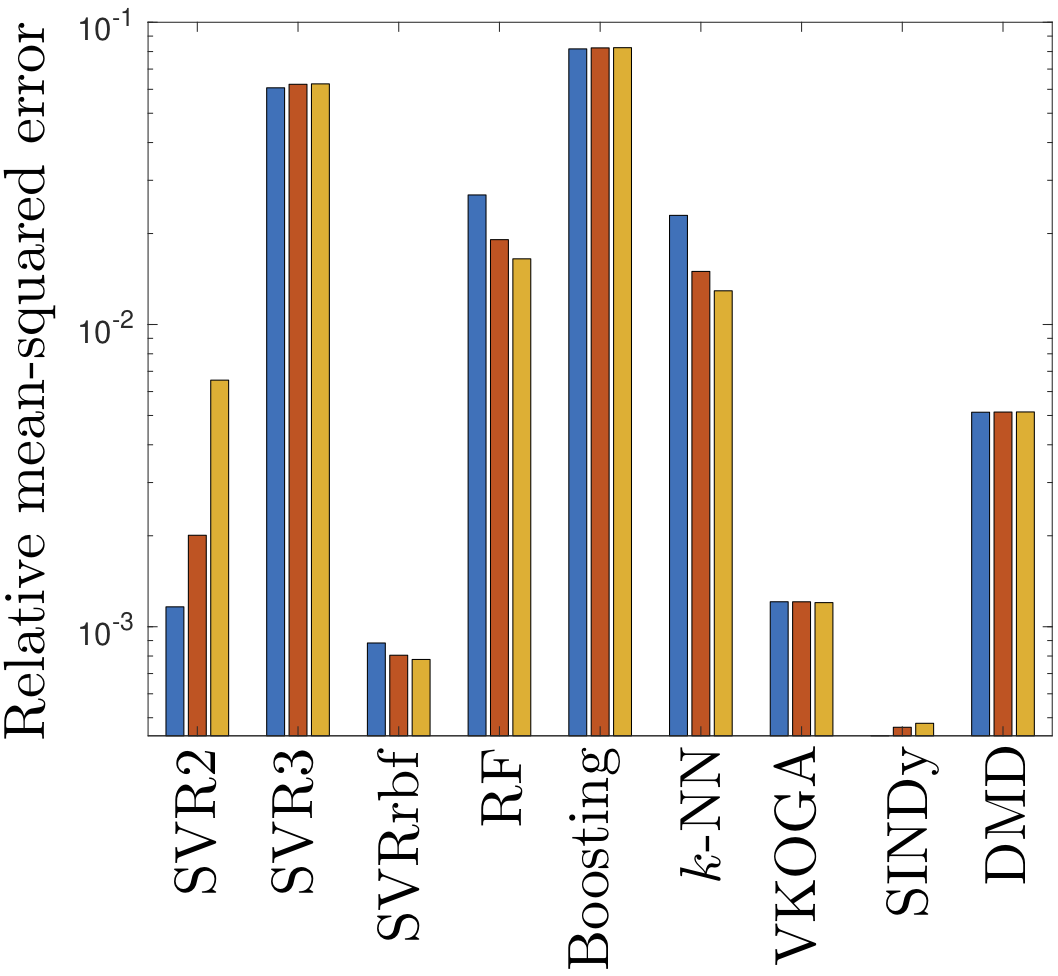}
			\caption{Training error}
  \end{subfigure}
    \begin{subfigure}[b]{0.48 \linewidth}
  \centering
		\includegraphics[width=\textwidth]{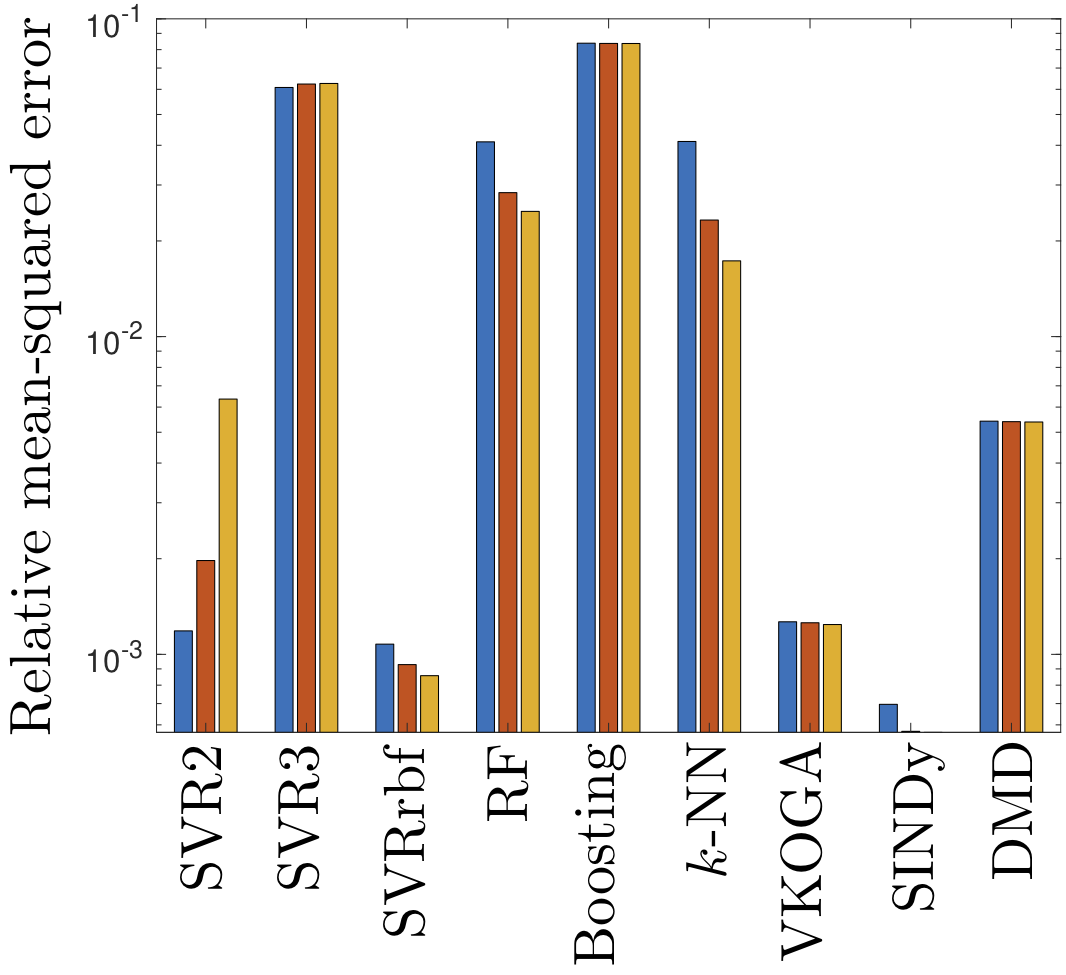}
			\caption{Testing error}
  \end{subfigure}
		\caption{Training and testing relative MSE errors for all regression methods. We
		consider models constructed using three training sets: one with the
		maximum training set $\timeDomainTrainRegression$ with
		$\card{\timeDomainTrainRegression} = 6\,450$ (yellow bars), one with 75\%
		of the maximum training data selected at random (\reviewer{red} bars), and one with
		50\% of the maximum training data selected at random (blue bars).}
	\label{fig:trainingTestingError}
	\end{center}
\end{figure}

First, note that the training and testing errors are extremely similar; this
suggests that the training data are indeed representative of the testing data,
indicating that we have not overfit our models. 

Second, note that even with the smallest training set characterized by only
50\% of the maximum number of considered points training set (i.e., $3\,225$ time
instances), the training and test errors are quite similar to those achieved
with the maximum training set. Indeed, the only regression methods where
doubling the training set yields improvements are  random forest, and $k$-NN.
Even in these cases, the additional training data provide modest gains. In the
case of SVR2, including additional training data actually \textit{degrades}
both the training and testing errors. This likely occurs because the SVR2
model is trained using the hinge loss (see Eq.~\eqref{eq:SVRopt}), which is
different from the relative MSE loss reported here. The hinge loss
does not penalize small errors, which do influence the relative MSE.

Third, note that the smallest errors (relative MSE on the test set $<1\times
10^{-3}$) are achieved by SVRrbf, VKOGA, and SINDy. We expect these
techniques to perform best when attempting to recover the desired sequence of
states by simulating the approximated discrete-time dynamics
\eqref{eq:discTimeApprox}. However, as discussed in Section
\ref{sec:boundedness}, only VKOGA and SVRrbf  are bounded; SINDy admits
unbounded dynamics.

\subsubsection{Approximated dynamics} 

We now turn to the final stage of the method: computing the sequence of
low-dimensional approximated states
$\{\stateGlobalVectorApproxArg{n}\}_{n=0}^\ntime$ according to approximated
discrete-time dynamics \eqref{eq:discTimeApprox}, and subsequently computing
the approximation to the desired sequence of states, i.e.,
$\{\stateGlobalVectorApproxArg{n}\}_{n=0}^\ntime$, where
$\stateGlobalVectorApproxArg{0} = \stateGlobalVectorArg{0}$ and
$\stateGlobalVectorApproxArg{n} = \redToFull(\stateApproxArg{n})$,
$n=1,\ldots,\ntime$. 

Figure \ref{fig:SimRandom8600} reports these results. For each
regression model (constructed with the maximum training set
$\timeDomainTrainRegression $), this figure reports the time-instantaneous
relative error 
$\|\stateApproxArg{n}-\stateArg{n}\|_2/\|\stateArg{n}\|_2$
over all time instances, i.e., for $n=0,\ldots,\ntime$. These
results show that both VKOGA  and SVRrbf yield accurate responses over time, as
the time-instantaneous relative error is less than 5\% and 3\% for these
models, respectively, at all time instances. These results are sensible, as
their training and testing errors were also small, as reported in Figure
\ref{fig:trainingTestingError}. Somewhat surprisingly, note that
SINDy---despite yielding small training and testing errors in Figure
\ref{fig:trainingTestingError}---produces an unstable response when integrated
into the approximated discrete-time dynamics \eqref{eq:discTimeApprox}. This
can be explained from the boundedness analysis in Section
\ref{sec:boundedness}: the polynomial basis functions employed with the SINDy
method are unbounded. Thus, this approach does not ensure a bounded
sequence of approximated states, and---in this case---produces unstable
discrete-time dynamics as a result. This highlights that the stability
properties of the ultimate dynamical system should play an important role in
deciding on the appropriate regression model; low training and testing errors
are not sufficient for an accurate model in this context.

				\begin{figure}[H]
					\centering
					\includegraphics[width=0.55\textwidth]{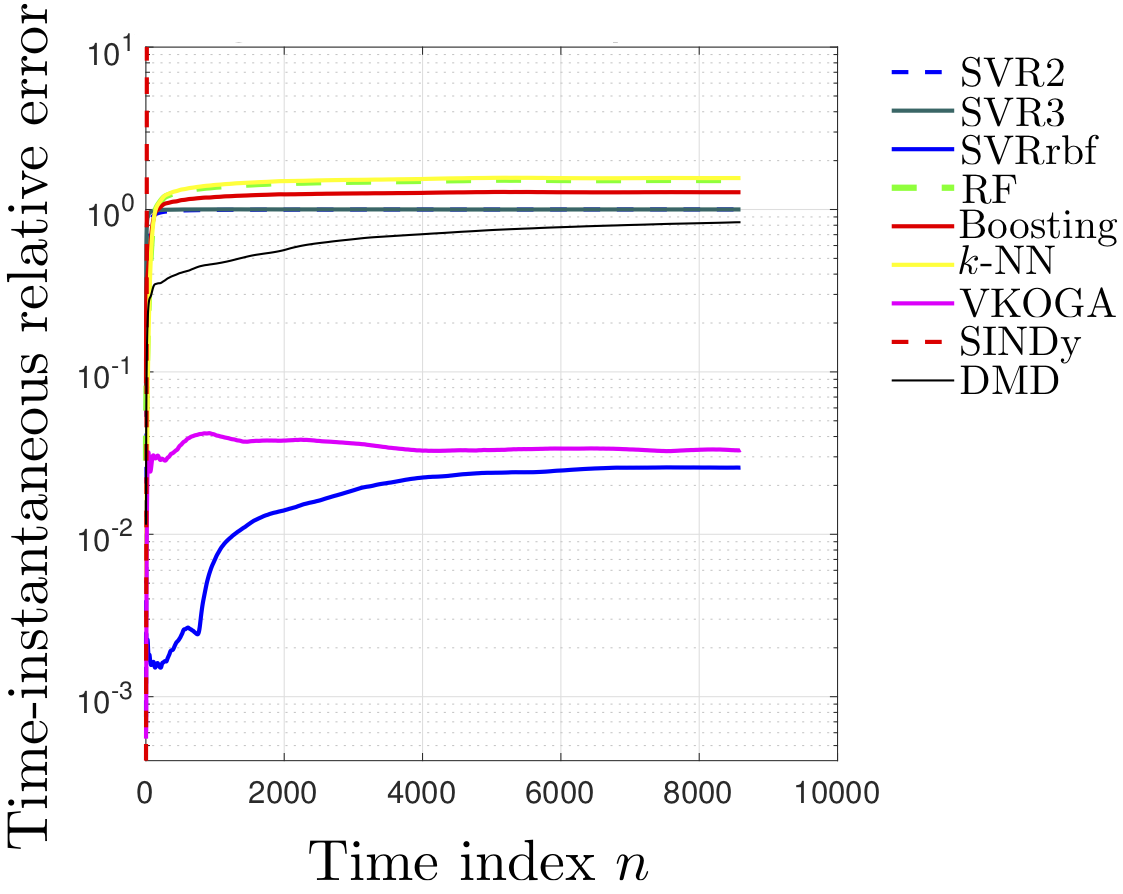}
					\caption{The time-instantaneous relative error 
					$\|\stateApproxArg{n}-\stateArg{n}\|_2/\|\stateArg{n}\|_2$,
					$n=1,\ldots,\ntime$
					over time for each
					considered regression model. Regression models are trained using the
					maximum training set $\timeDomainTrainRegression $. The
					low-dimensional state dimension is $\nstate=100$. Note
					that both SVRrbf and VKOGA yield accurate responses, which are
					consistent with their small training and testing errors. In
					contrast, SINDy produces an unstable response despite its small
					training and testing errors; this is due to the fact that it yields
					an unbounded approximated velocity.}
						\label{fig:SimRandom8600}
				\end{figure}

				We emphasize the promise of these results: the original sequence of
				data---wherein each vector has dimension
				$\nstateGlobalVector{\approx}13 \times 10^6$---has been accurately
				approximated as dynamical system of dimension $\nstate=100$, which
				corresponds to a compression ratio of $130\,000:1$. This implies
				that---with access only to the initial state
				$\stateGlobalVectorArg{n}$, the proposed method can recover the
				desired sequence of states to within 5\% accuracy in the
				low-dimensional state with knowledge of only the restriction operator
				$\fullToRed$ and approximated velocity $\velocityApprox$.

			We now consider increasing the dimension of the low-dimensional state to
			gauge the effect of accuracy on this parameter. To achieve this, after
			computing the principal components as described in Section
			\ref{sec:PCAexperiments}, we now preserve the first $\nstate=500$
			principal components, yielding a compression ratio of $26\,000:1$
			wherein we have captured 98\% of the statistical energy.
			Rather than repeat the experiment for all considered regression
			methods, we now only consider the VKOGA method, as it (along with
			SVRrbf) yielded the best results in the case of $\nstate=100$. Figure
			\ref{fig:SimRandom8600_k500} reports these results; comparing this
			figure with Figure \ref{fig:SimRandom8600} shows that increasing the
			dimension of the low-dimensional state from $\nstate=100$ to
			$\nstate=500$ 
			\reviewer{(and thus increasing the captured statistical energy in the
			global PCA step from 95\% to 98\%)}
			has reduced the relative error by nearly an order of
			magnitude. The response associated with $\nstate=500$ now exhibits
			time-instantaneous relative errors below $0.7\%$ over all time.

			\begin{figure}
				\centering
					\includegraphics[width=0.5\textwidth]{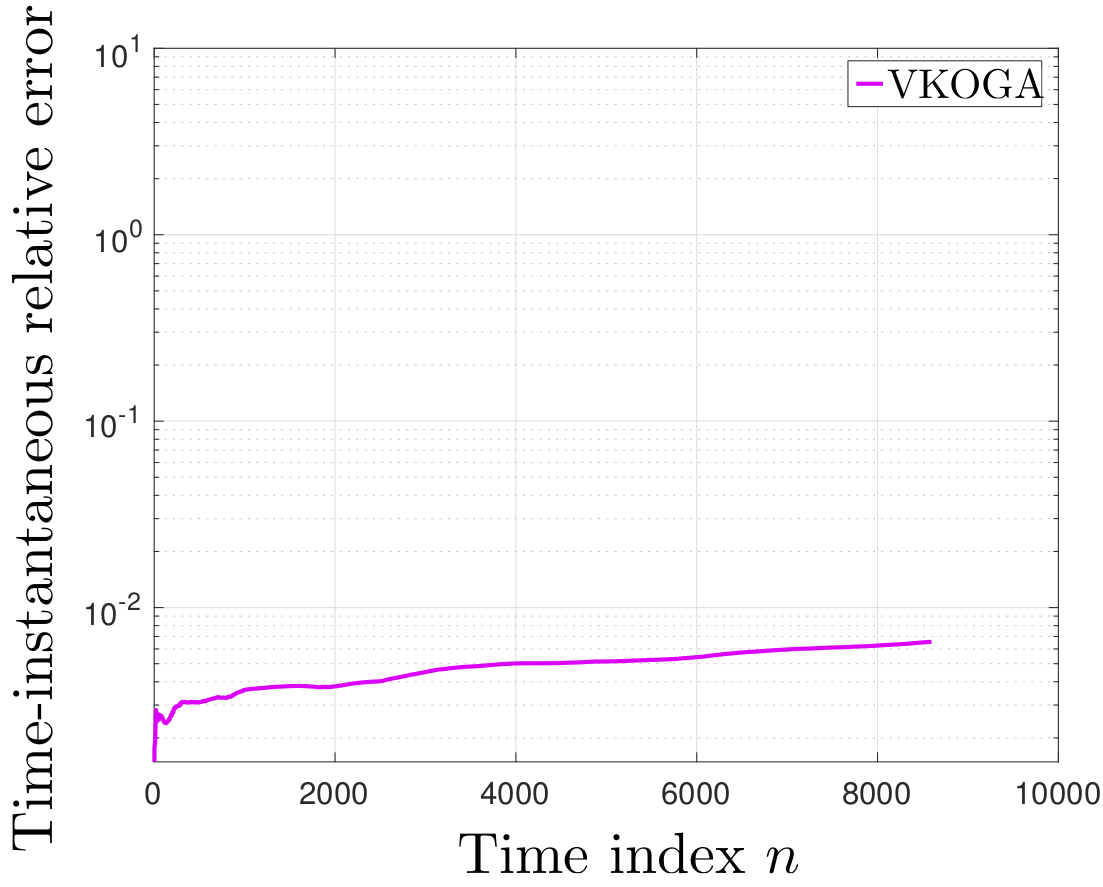}
					\caption{The time-instantaneous relative error 
					$\|\stateApproxArg{n}-\stateArg{n}\|_2/\|\stateArg{n}\|_2$,
					$n=1,\ldots,\ntime$
					for VKOGA
					trained using the
					maximum training set $\timeDomainTrainRegression $. We now employ a
					low-dimensional state dimension of $\nstate=500$. Note
					that the errors have decreased by roughly one order of magnitude
					relative to the case of $\nstate=100$ (compare with Figure \ref{fig:SimRandom8600}).}
						\label{fig:SimRandom8600_k500}
			\end{figure}

\begin{figure}
\centering
    \includestandalone[mode=image, width=0.48\linewidth]{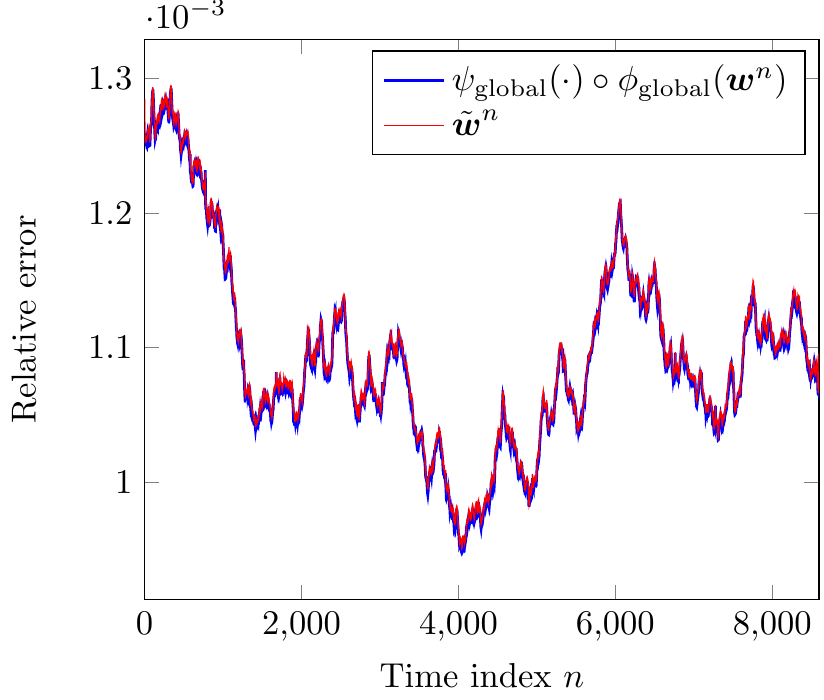}
    \caption{The time-instantaneous relative error for the autoencoder
		alone $\|\stateGlobalVectorArg{n} -
		\decoderFunctionGlobal(\cdot)\circ\encoderFunctionGlobal(\stateGlobalVectorArg{n})\|_1/\|\stateGlobalVectorArg{n}\|_1$,
		$n=1,\ldots,\ntime$,
		as well as for the approximated solutions $\|\stateGlobalVectorArg{n} -
		\stateGlobalVectorApproxArg{n}\|_1/\|\stateGlobalVectorArg{n}\|_1$,
		$n=1,\ldots,\ntime$. The $L_1$ norm is used because numerical roundoff errors cause spikes in time instances when $\|\stateGlobalVectorArg{n}\|_2$ is small.}
    \label{fig:error}
\end{figure}

The above results present the prediction error over time in the reduced state. 
Of greater interest is the prediction error relative to the global state
$\stateGlobalVectorArg{n}$ once the reduced state values are passed through
the prolongation operator $\redToFull$.
The prediction error across all time instances for the approximated solutions $\stateGlobalVectorApproxArg{n}$ can be found in Figure \ref{fig:error}.
To put these results in context, we compare against the error obtained by passing each global state through the encoder and decoder without performing the PCA and dynamics-learning procedures.
We note that the error introduced by compressing and reconstructing the global states through the autoencoder introduces very little error, with a relative error of less than $0.13\%$ across all time instances.
More critically, we can see that the relative error for the approximated states is nearly identical to the error from the autoencoder reconstructions, meaning that very little error is introduced by the PCA and dynamics-learning steps.

Beyond solely considering the global error in reconstructions, we are also interested in whether local properties of the fluid flow are preserved during the various stages of dimensionality reduction, dynamics propagation, and reconstruction.
One manner of determining whether flow properties are preserved is to consider the aggregate lift and drag forces that act on the surface of the half-cylinder over time.
By extracting pressure values on the surface of the cylinder, we can resolve the distribution of forces acting on the cylinder in the downstream and cross-stream directions.
Integrating over the surface of the half-cylinder, we obtain the lift and drag coefficients for the cylinder.
The left column of Figure \ref{fig:lift_drag} reports these lift and drag forces
plotted across all time instances for the CFD solutions
$\stateGlobalVectorArg{n}$, autoencoder reconstructions
$\decoderFunctionGlobal(\cdot)\circ\encoderFunctionGlobal(\stateGlobalVectorArg{n})$,
and approximated solutions $\stateGlobalVectorApproxArg{n}$ for
$n=1,\ldots,\ntime$.
The right column of Figure \ref{fig:lift_drag} reports the error in lift and
drag forces associated with the reconstructed solutions generated by the
autoencoder and the approximated solutions.
The absolute error is provided for lift values rather than the relative error due to some lift values being close to zero.
From these results, we first note that the error in lift and drag introduced by
encoding and decoding solutions with the autoencoder is quite low:
the absolute error is less than $1\e{-3}$ across all time steps for lift
values and the relative error is less than $0.1\%$ for drag values.
Furthermore, we note that the principal component analysis and
dynamics-learning procedures introduce very little additional error in these
quantities beyond what is introduced by the autoencoder.

We now provide several visualizations that lend qualitative insight into the
performance of the method.  Figure \ref{fig:reconstruction} compares instantaneous iso-surfaces of Q-criterion colored by velocity magnitude for the CFD solution, the autoencoder reconstruction, and the approximated solution.
The figure shows that for all considered cases, the autoencoder reconstruction
retains all of the key vortical structures of the CFD solution, albeit with
some noise in the iso-surfaces.  We also see that there is no distinguishable
difference between the autoencoder reconstruction and the approximated
solution.  This further supports our assertion that the principal component
analysis and dynamics-learning procedures introduce very little additional
error.

\begin{figure}[t]
\centering
    \begin{subfigure}[b]{0.48 \linewidth}
  \centering
    \includestandalone[mode=image, width=0.96\linewidth]{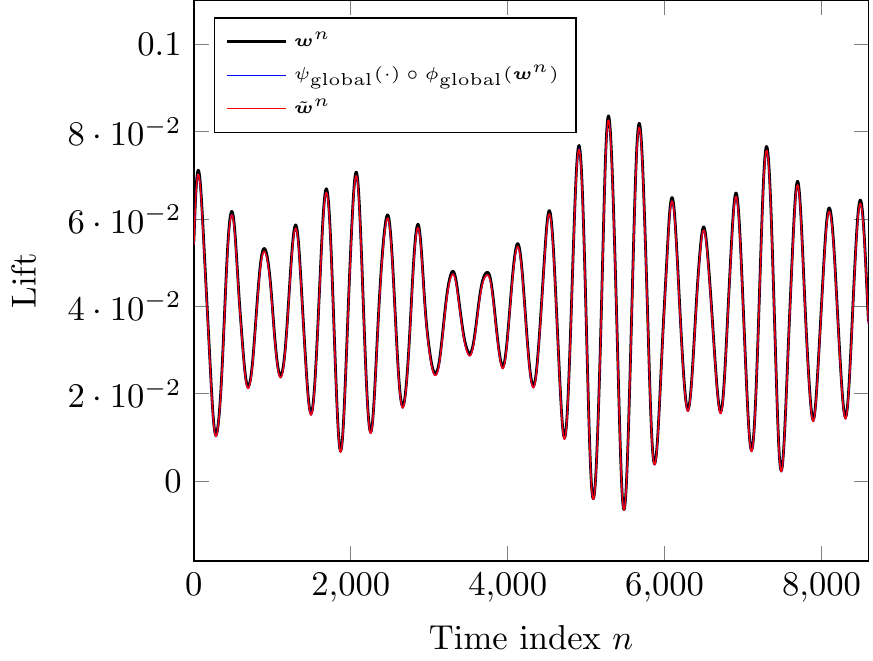}
    \caption{Lift predictions}
    \end{subfigure}
  \hfill
   \begin{subfigure}[b]{0.48 \linewidth}
   \centering
    \includestandalone[mode=image, width=0.9\linewidth]{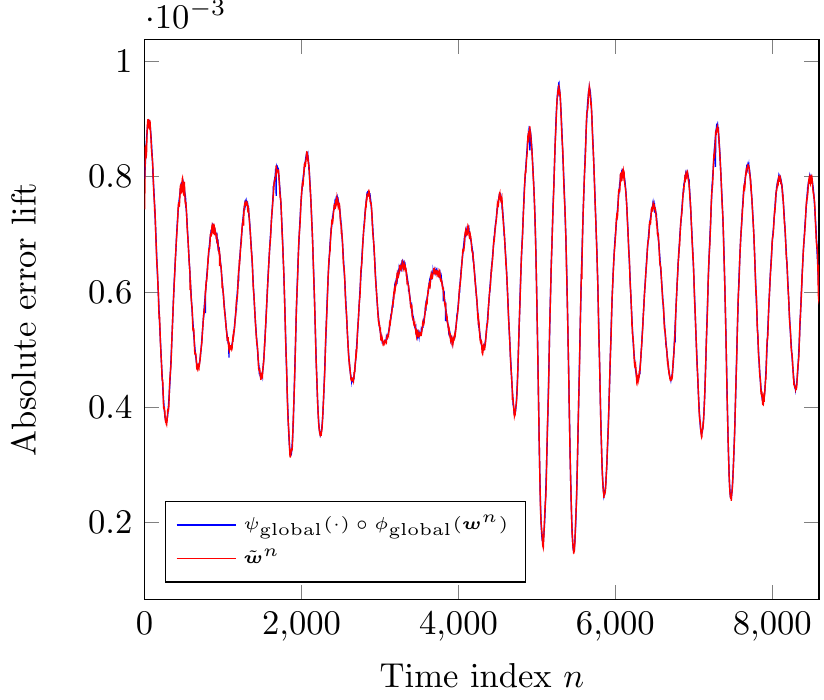}
	\caption{Lift prediction error}
  \end{subfigure}\vspace{1em}

  \begin{subfigure}[b]{0.48 \linewidth}
  \centering
    \includestandalone[mode=image, width=0.9\linewidth]{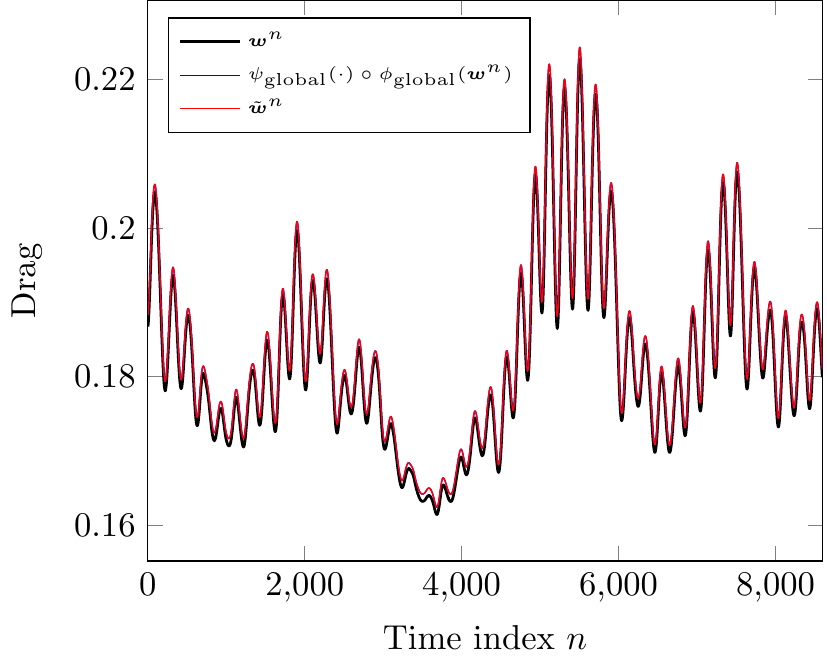}
    \caption{Drag predictions}
    \end{subfigure}
  \hfill
   \begin{subfigure}[b]{0.48 \linewidth}
  \centering
    \includestandalone[mode=image, width=0.9\linewidth]{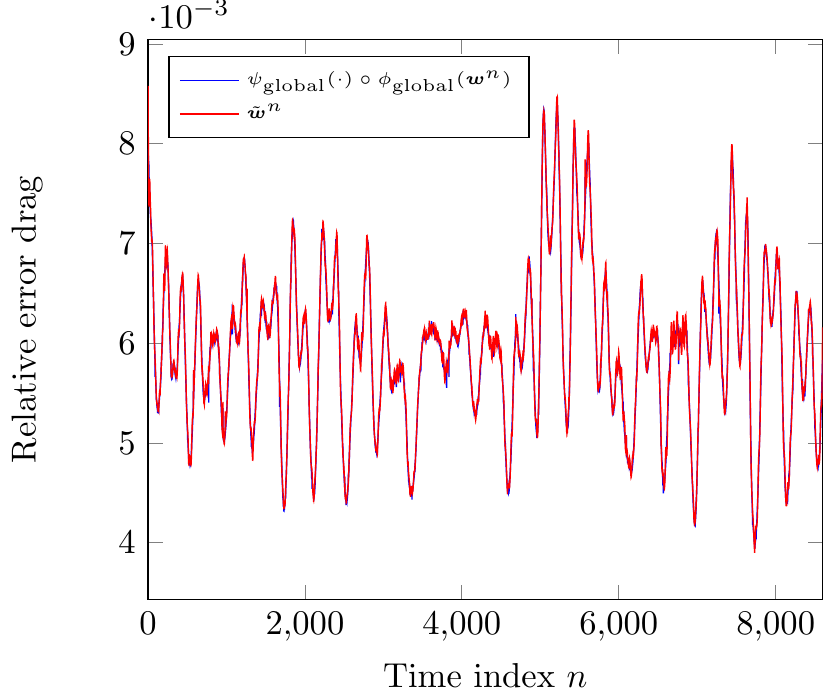}
    \caption{Drag prediction error}
  \end{subfigure}
	\caption{Lift and drag forces on the surface of the cylinder across all time instances. The left column shows lift and drag predictions, while the right column presents the error in lift and drag predictions. Absolute error is used rather than relative error to avoid numerical issues when lift values are close to zero.}
  \label{fig:lift_drag}
\end{figure}

\begin{figure}[t]
\centering
    \begin{subfigure}[b]{0.96 \linewidth}
  \centering
    \includegraphics[width=0.32\linewidth]{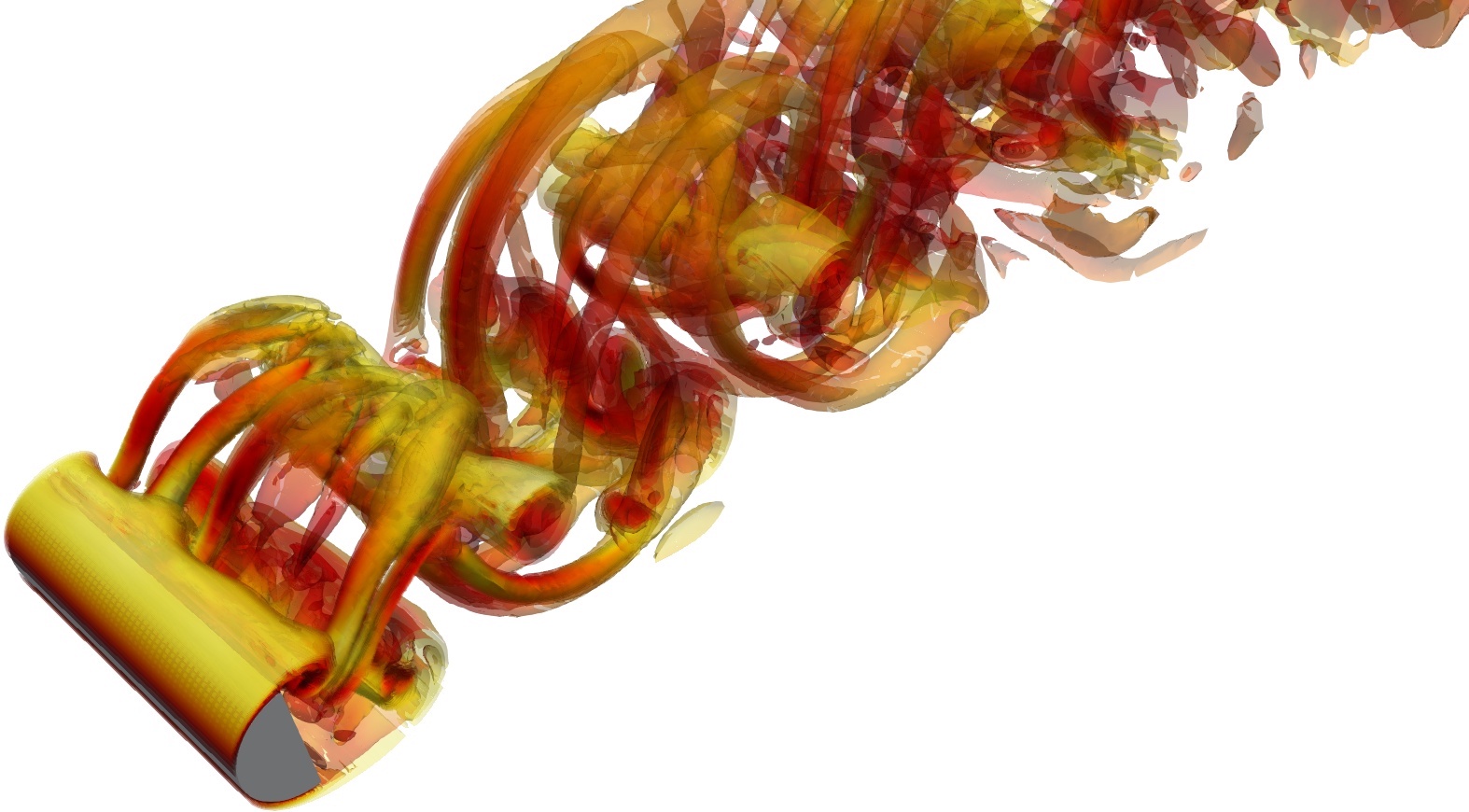}
  \hfill
    \includegraphics[width=0.32\linewidth]{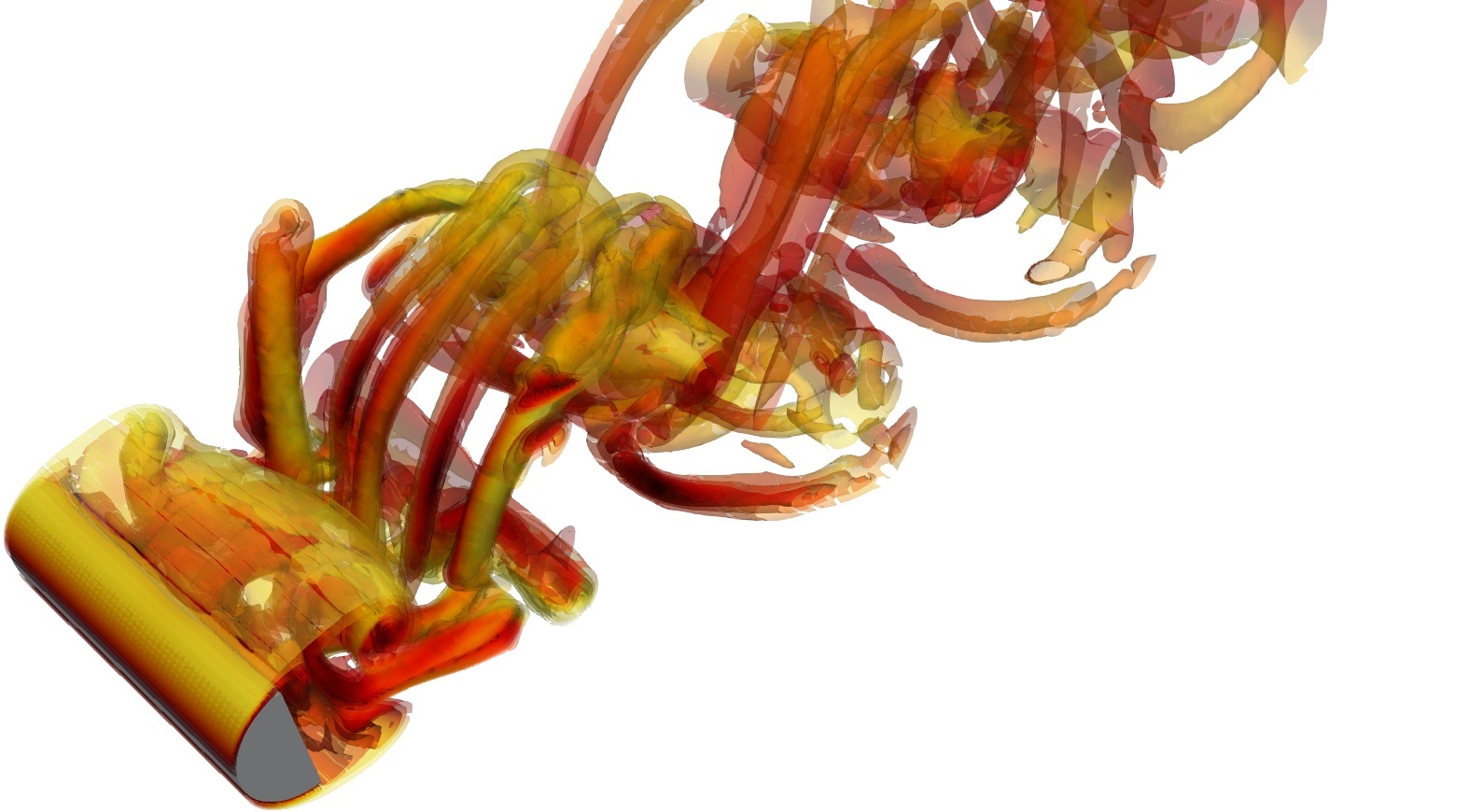}
    \hfill
    \includegraphics[width=0.32\linewidth]{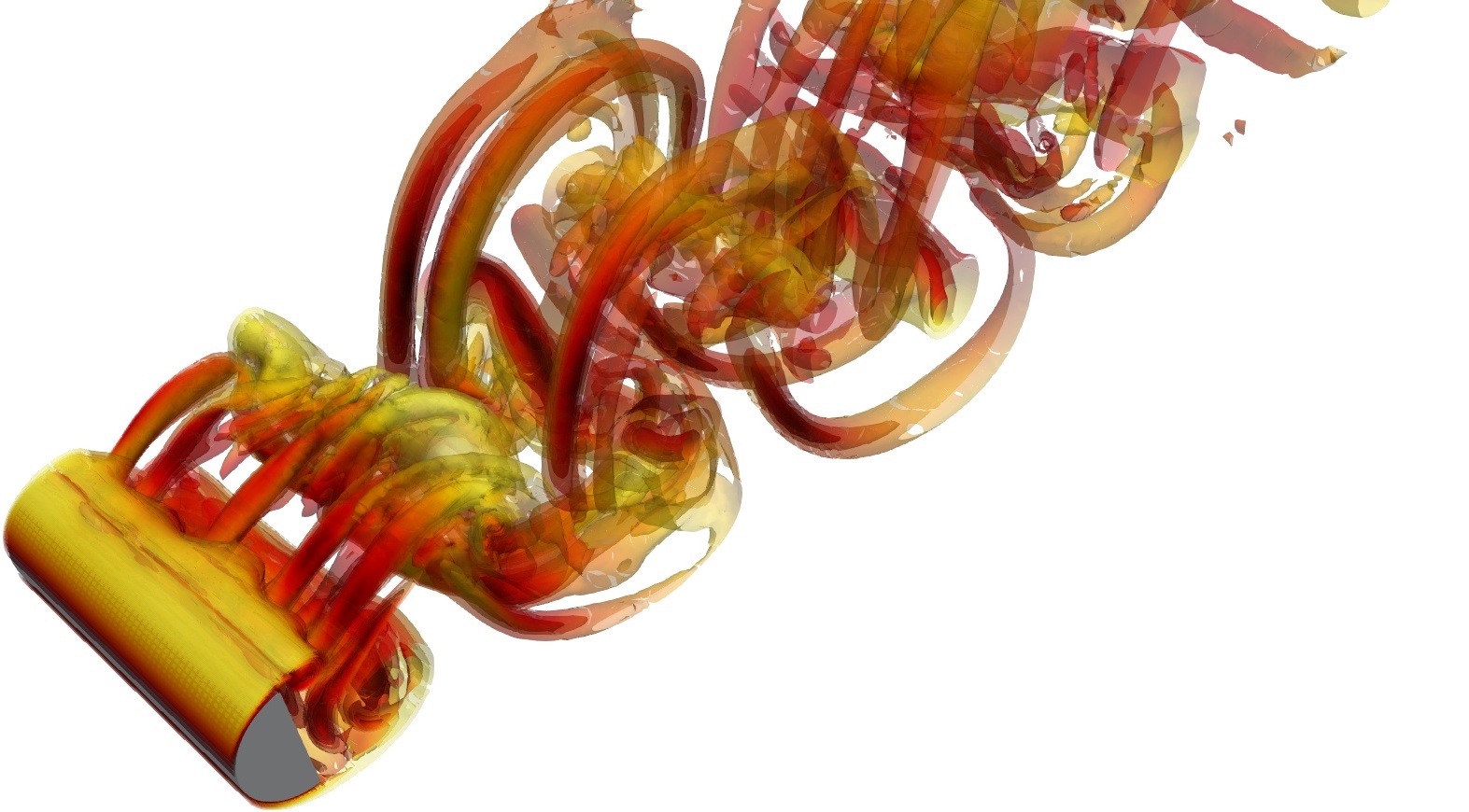}
			\caption{CFD solution $\stateGlobalVectorArg{n}$}\vspace{1em}
  \end{subfigure}
  \begin{subfigure}[b]{0.96 \linewidth}
  \centering
    \includegraphics[width=0.32\linewidth]{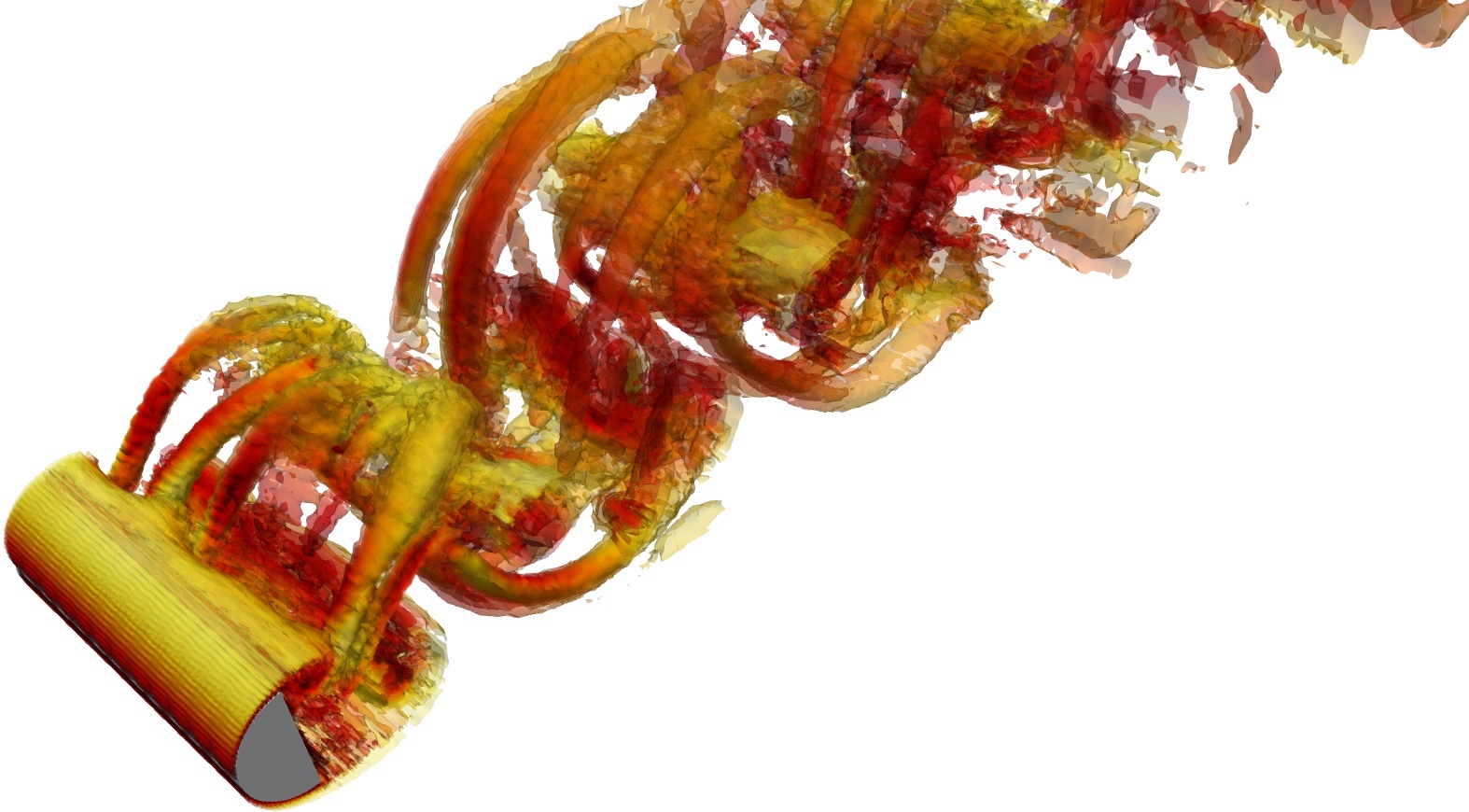}
  \hfill
    \includegraphics[width=0.32\linewidth]{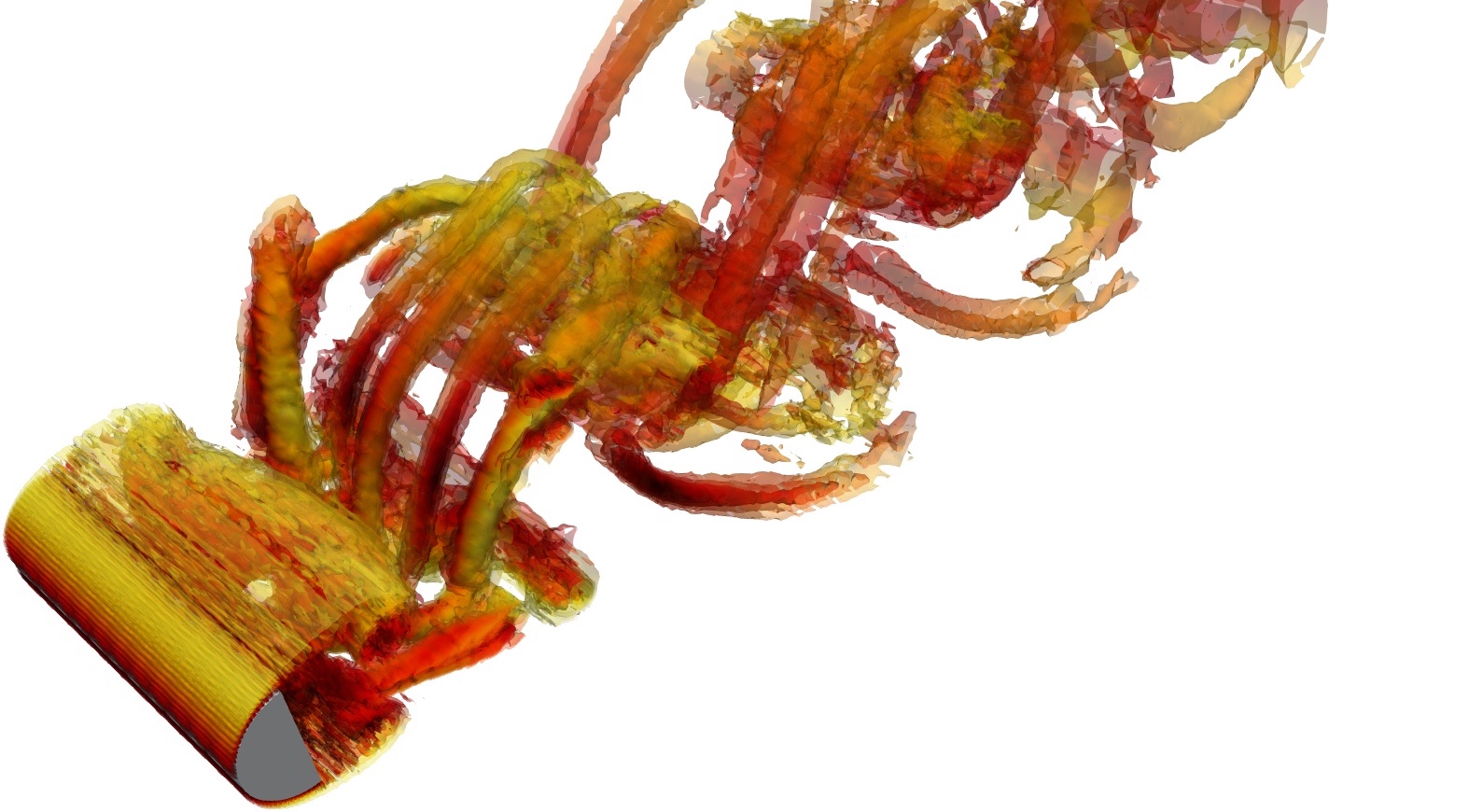}
    \hfill
    \includegraphics[width=0.32\linewidth]{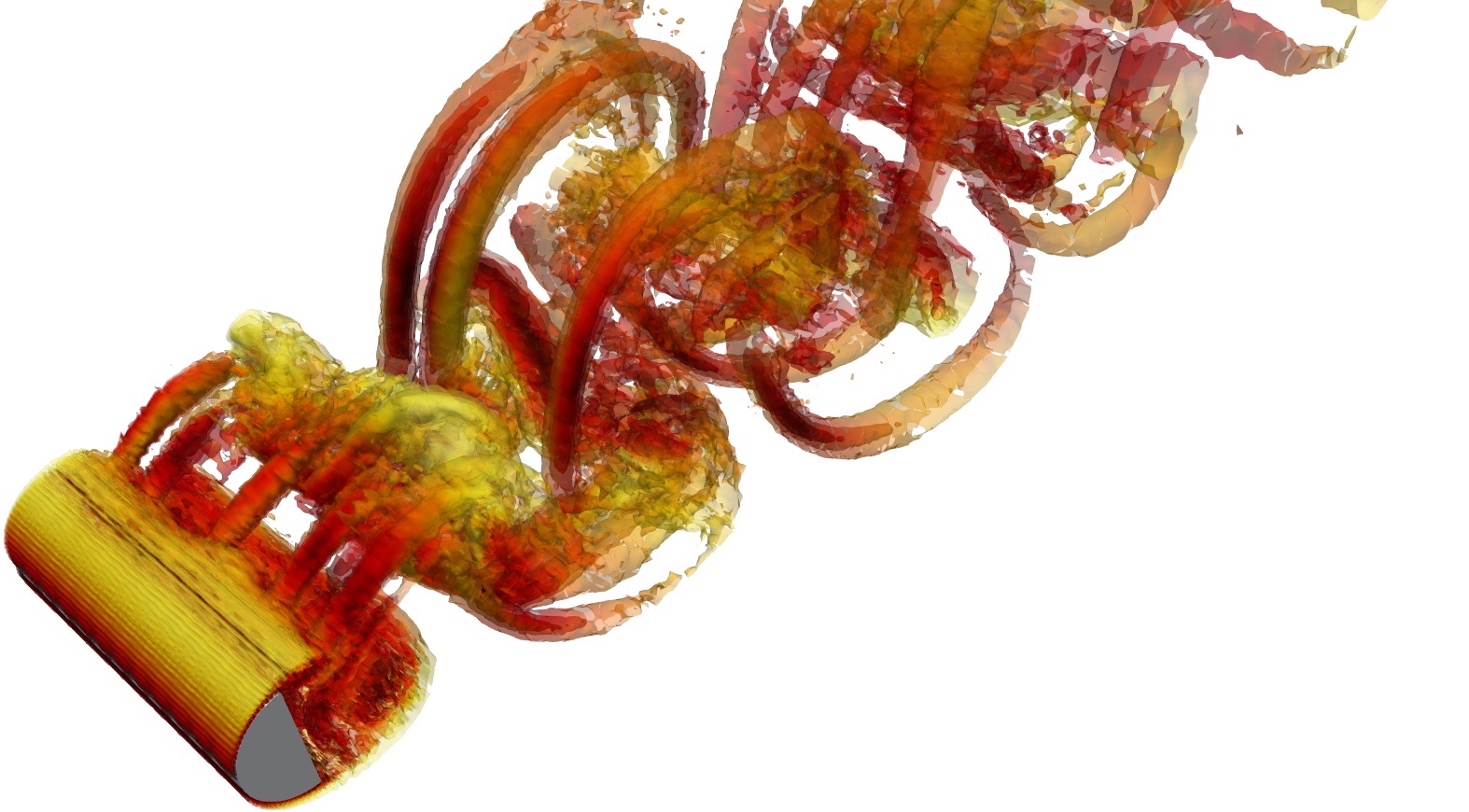}
			\caption{Autoencoder reconstruction $
		\decoderFunctionGlobal(\cdot)\circ\encoderFunctionGlobal(\stateGlobalVectorArg{n})$}\vspace{1em}
  \end{subfigure}
  \begin{subfigure}[b]{0.96 \linewidth}
  \centering
    \includegraphics[width=0.32\linewidth]{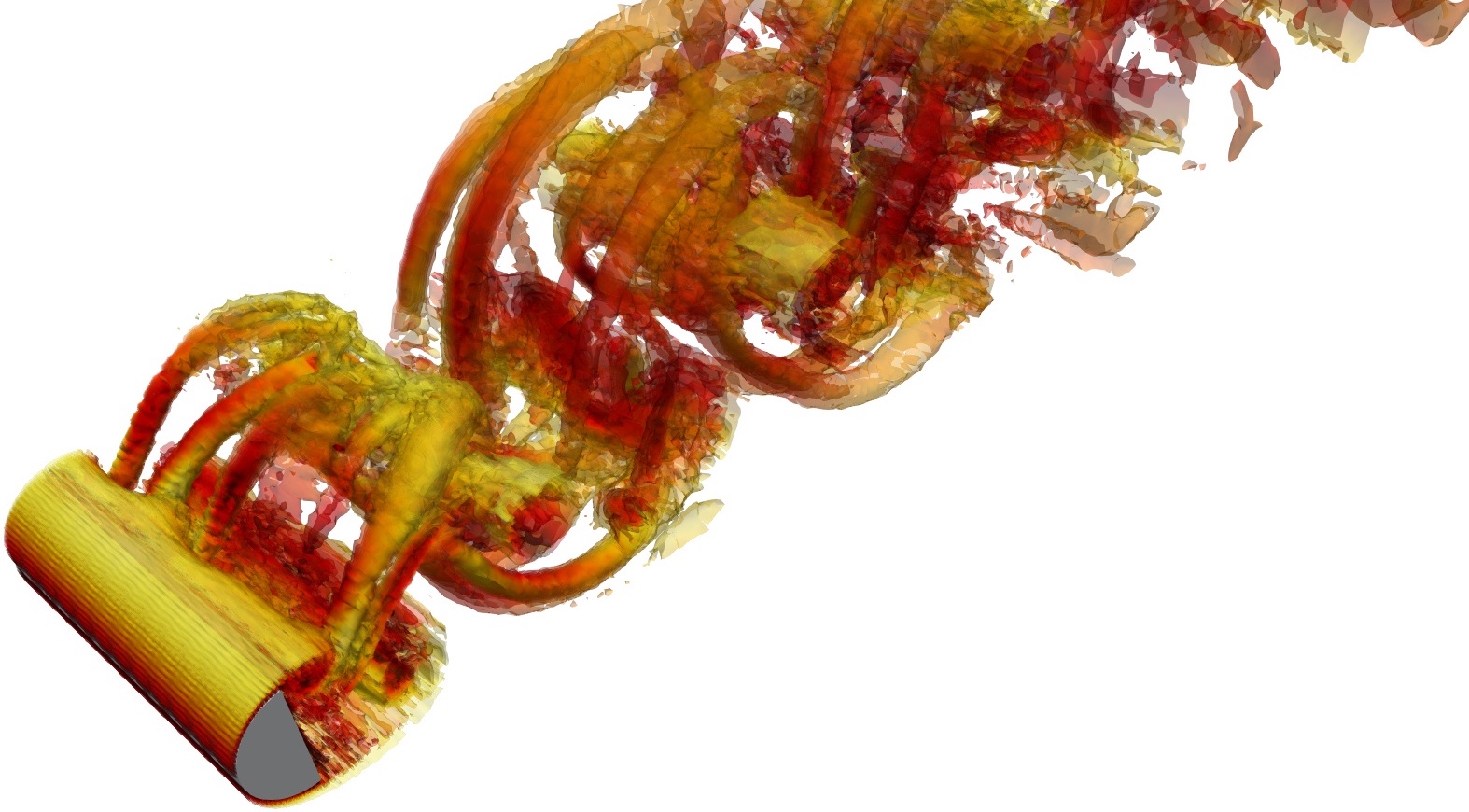}
  \hfill
    \includegraphics[width=0.32\linewidth]{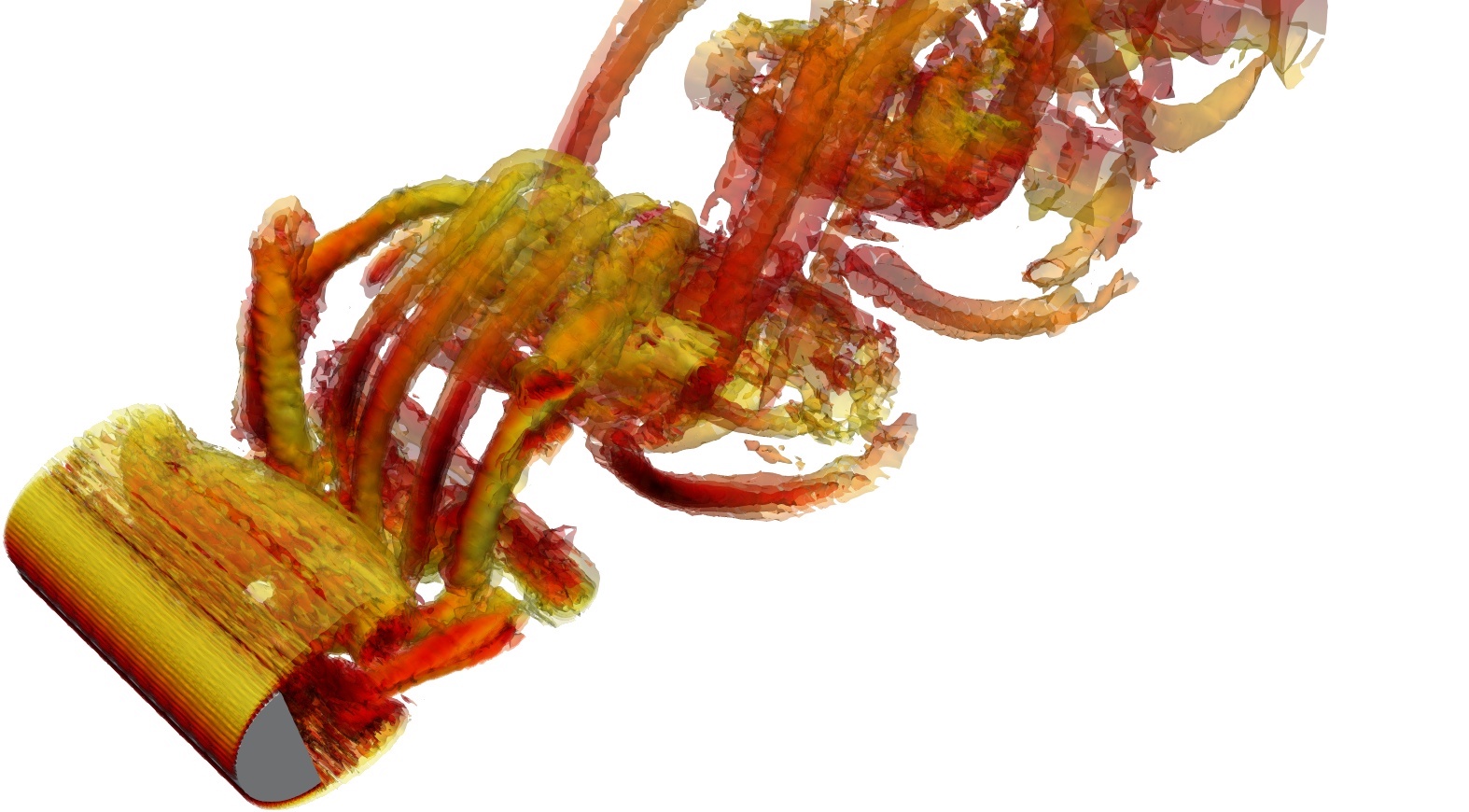}
    \hfill
    \includegraphics[width=0.32\linewidth]{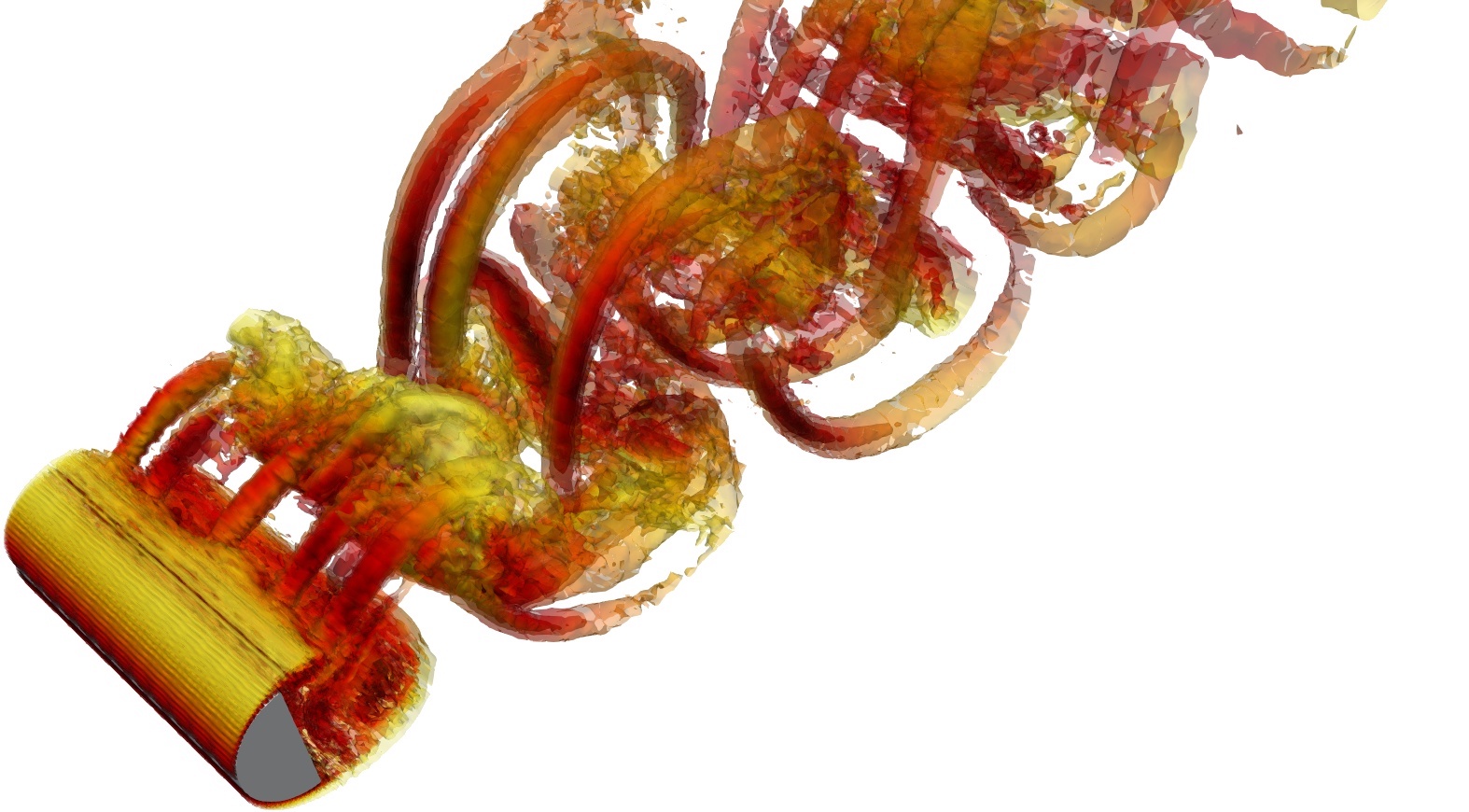}
			\caption{Approximated solution $\stateGlobalVectorApproxArg{n}$}\vspace{1em}
  \end{subfigure}

	\caption{Instantaneous iso-surfaces of Q-criterion colored by velocity
	magnitude at $n = 1\,000, 4\,000$ and \num{7\,000}.}
  \label{fig:reconstruction}
\end{figure}

\section{Conclusions}\label{sec:conclusions}

This article presented a novel methodology for recovering missing CFD data
given that the solution has been saved to disk at a relatively small number of
time steps.

The first stage of the methodology---hierarchical dimensionality
reduction---comprises a novel two-step technique tailored for high-order
discretizations. In the first step, we apply autoencoders for local compression; this
step computes a low-dimensional, nonlinear embedding of the degrees of freedom within
each element of the high-order discretization. In the numerical experiments,
we employed a convolutional neural network for this purpose. Results showed
that the autoencoder reduced the dimensionality of the element-local state
from $\nstateLocalVector=320$ to $\nstateLocalVectorRed = 24$ without
incurring significant errors (see Figures \ref{fig:lift_drag} and
\ref{fig:reconstruction}). In the second dimensionality-reduction step, we
apply principal component analysis to compress the global vector of encodings.
Results showed that this second step was able to reduce the number of global
degrees of freedom from $\nstateLocalVectorRed\nelements= 24\times 40\,584 \sim 10^6 $
to only 500, constituting a compression ratio of $26\,000:1$, while retaining very high levels of accuracy.

The second stage of the methodology---dynamics learning---applied regression
methods from machine learning to learn the discrete-time dynamics of the
low-dimensional state. We considered a wide range of regression methods for
this purpose, and found that support vector regression with a
radial-basis-function kernel (SVRrbf) and the vectorial kernel orthogonal
greedy algorithm (VKOGA) yielded the best performance (see Figure
\ref{fig:SimRandom8600}). Although the sparse identification of nonlinear
dynamics (SINDy) yielded low regression errors on an independent test set (see
Figure \ref{fig:trainingTestingError}), it did not produce an accurate
solution approximation due to its unboundedness; indeed, the resulting
trajectory was unstable (see Figure \ref{fig:SimRandom8600}). 

Ultimately, applying the proposed methodology with VKOGA and a low-dimensional
state dimension of 500 satisfied the objective of this work, as it enabled the
original CFD data to be reconstructed with extremely high levels of accuracy,
yielding low errors in the state, lift, and drag, as well as solution fields
that match well visually (see Figures \ref{fig:error}, \ref{fig:lift_drag},
and \ref{fig:reconstruction}). 

Future work involves introducing parameterization into the problem setting
such that the reconstruction task can be performed across multiple CFD
simulations characterized by different parameters, e.g., different boundary
conditions, geometric shapes, and operating conditions.

\section*{Acknowledgments}
K.~Carlberg's  and L.~Peng's research was sponsored by Sandia's Advanced
Simulation and Computing (ASC) Verification and Validation (V\&V) Project/Task
\#65755/002.01.19. F.~D.~Witherden's and A.~Jameson's research was supported
by the Air Force Office of Scientific Research via grant FA9550-14-1-0186.
J.~Morton's research was supported by the National Science Foundation Graduate
Research Fellowship Program under Grant No. DGE-114747.

\section*{References}
\bibliography{ref2}
\bibliographystyle{siam}

\end{document}